\documentclass[sigconf]{acmart}

\usepackage{array,graphicx,epsfig,amssymb,amsfonts}
\usepackage{bbm,amsmath,epstopdf,algorithm,subfigure, multirow}
\usepackage[noend]{algpseudocode}
\usepackage{color}
\usepackage{float}
\usepackage{dblfloatfix}
\usepackage{url}
\usepackage{pifont}
\usepackage[framemethod=1]{mdframed}
\allowdisplaybreaks
\newcommand{\be}{\begin{equation}}
	\newcommand{\ee}{\end{equation}}
\newcommand{\bee}{\begin{eqnarray}}
	\newcommand{\eee}{\end{eqnarray}}
\newcommand{\bse}{\begin{subequations}}
	\newcommand{\ese}{\end{subequations}}
\newcommand{\nnb}{\nonumber}

\newcommand{\cmark}{\ding{51}}
\newcommand{\xmark}{\ding{55}}

\newcommand{\specialcell}[2][c]{%
	\begin{tabular}[#1]{@{}c@{}}#2\end{tabular}}

\setcopyright{rightsretained}
\begin{document}
\title{Network Utility Maximization under Maximum Delay \\ Constraints and Throughput Requirements}
%\titlenote{Produces the permission block, and copyright information}
%\subtitle{Extended Abstract}
%\subtitlenote{The full version of the author's guide is available as \texttt{acmart.pdf} document}
%\begin{comment}
\author{Qingyu Liu, Haibo Zeng}
\affiliation{
	%\institution{Dept. of ECE, Virginia Tech}
	\department{Electrical and Computer Engineering}
	\institution{Virginia Tech}
	%\city{Blacksburg}
	%\state{Virginia}
	%\country{USA}
	}
%\email{qyliu14@vt.edu}
%\author{Haibo Zeng}
%\affiliation{
	%\institution{Dept. of ECE, Virginia Tech}
	%\department{Computer Engineering}
	%\institution{Virginia Tech}
	%\city{Blacksburg}
	%\state{Virginia}
	%\country{USA}
%	}
%\email{hbzeng@vt.edu}
\author{Minghua Chen}
\affiliation{
	%\institution{Dept. of IE, CUHK}
	\department{Information Engineering}
	\institution{The Chinese University of Hong Kong}
	%\city{Hong Kong}
	%\country{China}
	}	



\begin{abstract}

We consider the problem of maximizing aggregate user utilities over a multi-hop network, subject to link capacity constraints, maximum end-to-end delay constraints, and user throughput requirements. A user's utility is a concave function of the achieved throughput or the experienced maximum delay. The problem is important for supporting real-time multimedia traffic, and is uniquely challenging due to the need of simultaneously considering maximum delay constraints and throughput requirements. We first show that it is NP-complete either (i) to construct a feasible solution strictly meeting all constraints, or (ii) to obtain an optimal solution after we relax maximum delay constraints or throughput requirements up to constant ratios. We then develop a polynomial-time approximation algorithm named \textsf{PASS}. The design of \textsf{PASS} leverages a novel understanding between non-convex maximum-delay-aware problems and their convex average-delay-aware counterparts, which can be of independent interest and suggest a new avenue for solving maximum-delay-aware network optimization problems. Under realistic conditions, \textsf{PASS} achieves constant or problem-dependent approximation ratios, at the cost of violating maximum delay constraints or throughput requirements by up to constant or problem-dependent ratios. \textsf{PASS} is practically useful since the conditions for \textsf{PASS} are satisfied in many popular application scenarios. We empirically evaluate \textsf{PASS} using extensive simulations of supporting video-conferencing traffic across Amazon EC2 datacenters. Compared to existing algorithms and a conceivable baseline, \textsf{PASS} obtains up to $100\%$ improvement of utilities, by meeting the throughput requirements but relaxing the maximum delay constraints that are acceptable for practical video conferencing applications.

\end{abstract}

%\begin{IEEEkeywords}
%	Delay-aware network flow, maximum delay optimization, video traffic, inter-datacenter networks, exact pseudo-polynomial time algorithm, approximate algorithm.
%\end{IEEEkeywords}

\copyrightyear{2018} 
\acmYear{2018} 
\setcopyright{acmcopyright}
\acmConference[Submission to MobiHoc '19]{The Twentieth International Symposium on Mobile Ad Hoc Networking and Computing}{July 2--5, 2019}{Catania, Italy}
\acmBooktitle{Submission to MobiHoc '19, July 2--5, 2019, Catania, Italy}
\acmPrice{15.00}
\acmDOI{x}
\acmISBN{x}

\maketitle

\section{Introduction}
\begin{table*}[]
	\centering
	\caption{Compare our work with existing studies.}
	\label{tab:relatedwork}
	\scalebox{0.9}{\begin{tabular}{|c||c|c|c|c|c|}
			\hline
			 \multirow{2}{*}{} & \multicolumn{2}{c|}{Maximization Objective} &  \multicolumn{2}{c|}{Constraints}  &  Networking Setting \\ \cline{2-6}
			  & \specialcell{Aggregate Throughput-\\Based Utilities} & \specialcell{Aggregate Maximum-\\Delay-Based Utilities} & \specialcell{Throughput\\Requirements} & \specialcell{Maximum Delay\\Constraints} & Multiple-Unicast \\ \hline\hline
			 Many, e.g.,~\cite{kelly1998rate,low1999optimization,wang2003can,palomar2006tutorial} & \cmark & \xmark & \cmark & \xmark &  \cmark \\ \hline
			 \cite{misra2009polynomial,zhang2010reliable,correa2004computational,correa2007fast,my} & \xmark & \cmark$^*$ & \cmark & \xmark  & \xmark \\ \hline
			 \cite{cao2017optimizing,yu2018application} & \cmark$^{**}$ & \xmark & \xmark & \cmark & \cmark \\ \hline
			 Out Work & \cmark & \cmark & \cmark & \cmark & \cmark \\ \hline
	\end{tabular}}
\linebreak
\footnotesize \emph{Note}. $^*$: The objective of~\cite{misra2009polynomial,zhang2010reliable,correa2004computational,correa2007fast,my} is to minimize maximum delay, which is a special case of maximizing maximum-delay-based utility functions.

$^{**}$: The objective of~\cite{cao2017optimizing,yu2018application} is to maximize throughput, which is a special case of maximizing throughput-based utility functions.
\end{table*}

We consider a multiple-unicast communication scenario where each unicast source streams a network flow to its destination over a multi-hop network, possibly using multiple paths. We study the problem of maximizing aggregate user utilities, subject to link capacity constraints, maximum delay constraints, and user throughput requirements. A user's utility is a concave function of the achieved throughput or the experienced maximum delay. The \textit{maximum delay} denotes the maximum Source-to-Destination (\textsf{S2D}) delay, or equivalently the delay of the slowest \textsf{S2D} path that carries traffic. 

Our study is motivated by the increasingly interests on supporting delay-critical traffic in various applications, e.g., video conferencing~\cite{chen2013celerity,liu2016delay,hajiesmaili2017cost}. It is reported that 51 million users per month attend WebEx meetings, and 3 billion minutes of calls per day use Skype~\cite{my}. Low \textsf{S2D} delay is vital for such video conferencing applications. As recommended by the International Telecommunication Union (ITU)~\cite{ITU}, a delay less than 150ms can provide a transparent interactivity while delays above 400ms are unacceptable for video conferencing. We remark that the maximum \textsf{S2D} delay, instead of the average one, is a critical concern for provisioning low delay services, since there may exist traffic which experiences an arbitrarily large \textsf{S2D} delay even for the solution that minimizes average \textsf{S2D} delay performance~\cite{my}. In sharp contrast, all the traffic can be streamed from its source to its destination timely following any solution that has an acceptable maximum \textsf{S2D} delay performance, because the maximum \textsf{S2D} delay is defined as an upper bound of \textsf{S2D} delays of all the traffic.

We consider a delay model where transmission over a link experiences a constant delay if the aggregate flow rate of the link is within a constant capacity, and unbounded delay otherwise. This model fits a number of practical applications, particularly the routing of delay-critical video conferencing traffic over inter-datacenter networks. 
%We note that using multiple paths is necessary for delay-sensitive services when the shortest path is insufficient to support the traffic demand due to the limited networking resources.
%
Specifically, according to recent reports from Microsoft~\cite{hong2013achieving} and Google~\cite{jain2013b4}, most real-world inter-datacenter networks are characterized by sharing link bandwidth for different applications, with over-provisioned link capacities. (i) Real-world inter-datacenter networks nowadays are utilized to simultaneously support traffic from various services, some of which have stringent delay requirements (e.g., video conferencing) while others are bandwidth-hungry and less sensitive to delay (e.g., data maintenance). Link capacity is often reserved separately for different types of services depending on their characteristics. (ii) Cloud providers typically over-provision inter-datacenter link capacity by $2-3$ times on a dedicated backbone to guarantee reliability, and the average link-capacity utilizations (the aggregate utilization of applications, not the bandwidth-utilization of individual applications) for busy links are $30-60\%$~\cite{liu2016delay}. As such, for applications whose traffic volume is within the reserved capacity for their types of service, queuing delays are negligible and the constant propagation delays dominate end-to-end delays, as evaluated by~\cite{liu2016delay} in a realistic network of Amazon EC2. Otherwise, if the traffic volume exceeds the reserved capacity, the applications will start to experience substantial queuing delays and thus substantial end-to-end delays. \textit{These observations justify our link capacity and delay model}, especially for the critical problem of routing video-conferencing traffic over real-world inter-datacenter networks.

\subsection{Existing Studies}
We summarize existing studies in Tab.~\ref{tab:relatedwork}. In the literature, there exist many network utility maximization studies with throughput concerns, e.g.,~\cite{kelly1998rate,low1999optimization,wang2003can,palomar2006tutorial}, but less of them consider maximum delays. This is because the maximum delay of a single-unicast network flow is non-convex with the flow decision variables, and hence even a maximum-delay-aware problem in a simple networking scenario, e.g., the single-unicast maximum delay minimization problem, is NP-hard and thus challenging to solve~\cite{misra2009polynomial}. 

Misra \emph{et al.}~\cite{misra2009polynomial} study the single-unicast maximum delay minimization problem subject to a throughput requirement, and design a Fully-Polynomial-Time Approximation Scheme (\textsf{FPTAS}). Zhang \emph{et al.}~\cite{zhang2010reliable} generalize the \textsf{FPTAS} of~\cite{misra2009polynomial} and develop an \textsf{FPTAS} to minimize maximum delay subject to throughput, reliability, and differential delay constraints also in the single-unicast scenario. We observe that both \textsf{FPTAS}es require to solve flow problems iteratively in time-expanded networks, by employing a binary-search based idea applicable only in the single-unicast setting. It is thus unclear how to extend their techniques to the general multiple-unicast scenario where the utility of an unicast (user) can be a concave function with the experienced maximum delay.

Cao \emph{et al.}~\cite{cao2017optimizing} develop an \textsf{FPTAS} that can maximize throughputs subject to maximum delay constraints in a multiple-unicast setting. This \textsf{FPTAS} is generalized by Yu \emph{et al.}~\cite{yu2018application} to design \textsf{FPTAS}es for other throughput maximization problems for practical \textsf{IoT} applications. Similar to \textsf{FPTAS}es proposed by~\cite{misra2009polynomial,zhang2010reliable}, to satisfy maximum delay constraints while optimizing throughputs, \textsf{FPTAS}es of~\cite{cao2017optimizing,yu2018application} require to solve flow problems iteratively in time-expanded networks, which is time-consuming. Moreover, the design of \textsf{FPTAS}es in~\cite{cao2017optimizing,yu2018application} leverages the primal-dual algorithm, where their primal problems and associated dual problems need to be casted as linear programs. It is unclear how to extend their technique to the general scenario where the utility of an unicast can be a concave function with the achieved throughput.

We note that there exist other maximum-delay-aware studies in the literature. However, they only develop heuristic approaches instead of approximation algorithms. For example, Liu \emph{et al.}~\cite{liu2016delay} target the multicast maximum delay optimization problems. Their heuristic approach suffers from two limitations: (i) the running time could be high because the number of variables increases exponentially in the network size, and (ii) there is not yet theoretical performance guarantee of the achieved solution.

Instead of modeling link delay as a constant within a capacity as in~\cite{misra2009polynomial,zhang2010reliable,cao2017optimizing,yu2018application,liu2016delay}, there exist studies which model the link delay as a link-flow-dependent function. For example, Correa \emph{et al.}~\cite{correa2004computational,correa2007fast} minimize maximum delay with delay-function-dependent approximation ratios guaranteed. Liu \emph{et al.}~\cite{my} minimize maximum delay with constant approximation ratios guaranteed. Our study models link delay as a constant within a capacity, which is the same as those in~\cite{misra2009polynomial,zhang2010reliable,cao2017optimizing,yu2018application,liu2016delay}, but different from the ones in~\cite{correa2004computational,correa2007fast,my}. We remark that maximum-delay-aware problems are fundamentally different with these different link delay models, since it is APX-hard to minimize the single-unicast maximum delay (hence no \textsf{PTAS} exists unless P = NP) with the flow-dependent delay model~\cite{correa2007fast}, but an \textsf{FPTAS}\footnote{Unless P = NP, it holds that $\textsf{FPTAS}\subsetneq\textsf{PTAS}$ in that the runtime of a \textsf{PTAS} is required to be polynomial in problem input but not $1/\epsilon$, while the runtime of an \textsf{FPTAS} is polynomial in both the problem input and $1/\epsilon$~\cite{WikiPTAS}.} exists to minimize the single-unicast maximum delay with the constant delay model~\cite{misra2009polynomial}.

Overall, with the constant delay model, existing maximum-delay-aware studies focus on either the throughput-constrained maximum delay minimization problem or the maximum-delay-constrained throughput maximization problem, which are just special cases of our problem (Tab.~\ref{tab:relatedwork}). To design approximation algorithms, they rely on a technique of solving problems in expanded networks iteratively, leading to impractically high time complexities (e.g., at least $O(|E|^3|V|^4\mathcal{L})$ to minimize single-unicast maximum delay where $|V|$ is number of nodes, $|E|$ is number of links, and $\mathcal{L}$ is input size of the given problem instance~\cite{misra2009polynomial}). It is unclear how to generalize their techniques to our multiple-unicast utility maximization scenario, where the utility of an unicast is a concave function of the achieved throughput or the experienced maximum delay. In sharp contrast, we develop an approximation algorithm for our problem of maximizing utilities, by leveraging a novel understanding between non-convex maximum-delay-aware problems and their convex average-delay-aware counterparts. Specifically, we solve an average-delay-aware problem only once in the input network, and then deletes certain flow rate from individual unicast flows, resulting in a small time complexity (e.g., $O(|E|^3\mathcal{L})$ to minimize single-unicast maximum delay in a dense network (Thm.~\ref{thm:sufficient-condition}). 

\subsection{Our Contributions}
In this paper, we study a multiple-unicast flow problem of maximizing aggregate user utilities over a multi-hop network, subject to link capacity constraints, maximum delay constraints, and user throughput requirements. We make the following contributions.

$\mbox{\ensuremath{\rhd}}$
We prove that it is NP-complete either (i) to construct a feasible solution meeting all constraints, or (ii) to obtain an optimal solution after we relax maximum delay constraints or throughput requirements up to constant ratios, due to the need of simultaneously considering maximum delay constraints and user throughput requirements. 

$\mbox{\ensuremath{\rhd}}$
We design an algorithm named \textsf{PASS} (Polynomial-time Algorithm Supporting utility-maximal flows Subject to throughput/delay constraints) for constructing approximate solutions to our problem in a polynomial time. Our design leverages a novel understanding between non-convex maximum-delay-aware problems and their convex average-delay-aware counterparts, which can be of independent interest and suggests a new avenue for solving maximum-delay-aware network optimization problems.

$\mbox{\ensuremath{\rhd}}$
We characterize sufficient conditions for \textsf{PASS} to solve our problem in a polynomial time, providing (i) a constant approximation ratio after relaxing throughput requirements and maximum delay constraints by constant ratios, or (ii) a problem-dependent approximation ratio satisfying maximum delay constraints, after relaxing throughput requirements by a problem-dependent ratio, or (iii) a problem-dependent approximation ratio satisfying throughput requirements, after relaxing maximum delay constraints by a problem-dependent ratio. We note that one can use pre-scaled maximum delay constraints or throughput requirements as the input to \textsf{PASS} to generate feasible solutions as the output. 

$\mbox{\ensuremath{\rhd}}$
We observe that our characterized conditions are satisfied in many popular application settings, where \textsf{PASS} can be applied with strong theoretical performance guarantee. Representative settings include minimizing throughput-constrained maximum delay and maximizing maximum-delay-constrained network utility. We evaluate the empirical performance of \textsf{PASS} in simulations of supporting video-conferencing traffic across Amazon EC2 datacenters. Compared to existing algorithms as well as a conceivable baseline, \textsf{PASS} can obtain up to $100\%$ improvement of utilities, by meeting throughput requirements but relaxing maximum delay constraints that are acceptable for video conferencing applications.

\section{System Model}\label{sec:system-model}
\subsection{Preliminary}
We consider a multi-hop network modeled as a directed graph $G \triangleq (V,E)$ with $|V|$ nodes
and $|E|$ links. Each link $e \in E$ has a constant capacity $c_e\ge 0$ and a constant delay $d_e\ge 0$. For each link $e\in E$, data streamed to $e$ experiences a delay of $d_e$ to pass it, and the rate of streaming data to $e$ must be within the capacity $c_e$. We are given $K$ users, where for each user $i$ ($i=1,2,...,K$), a source $s_i \in V$ needs to stream a single-unicast network flow to a destination $t_i \in V\backslash\{s_i\}$, possibly using multiple paths.

We denote $P_i$ as the set of all simple paths from $s_i$ to $t_i$, and $P\triangleq\cup_{i=1}^K P_i$. For any $p \in P$, its path delay $d^p$ is defined as
\be
d^p ~~\triangleq~~ \sum_{e\in E:e\in p} d_e, \nnb
\ee
i.e., the summation of link delays along the path. We denote a multiple-unicast network flow solution as $f \triangleq\{f_i,i=1,2,...,K\}$, where a single-unicast flow $f_i$ is defined as the assigned flow rate over $P_i$, i.e., $f_i \triangleq \{x^p: x^p \ge 0, p \in P_i\}$. For $f_i$, we define
\be
x_i^e~~ \triangleq~~ \sum_{p\in P_i: e \in p} x^p\nnb
\ee
as the aggregated link rate of $e\in E$ of the unicast $i$ (or the user $i$ equivalently). Similarly, we denote $x_e$ as the total aggregated link rate of link $e\in E$, and
\be
x_e~~\triangleq~~ \sum_{i=1}^{K}x_i^e~~=~~\sum_{p\in P:e\in p}x^p.\nnb
\ee
We further denote the flow rate, or the \textbf{throughput} equivalently, achieved by a single-unicast flow $f_i$ by $|f_i|$,
\be\label{eqn:rate}
|f_i|~~\triangleq~~\sum_{p\in P_i}x^p~~=~~\sum_{e\in\textsf{Out}(s_i)}x_i^e~~=~~\sum_{e\in\textsf{In}(t_i)}x_i^e,\nnb
\ee
where $\textsf{Out}(v)$ (resp. \textsf{In}(v)) is the set of outgoing (resp. incoming) links of $v$. The \textbf{maximum delay} experienced by $f_i$ is defined as
\be
%\label{eqn:maxdelay}
\mathcal{M}(f_i)~~ \triangleq~~ \max_{p \in P_i: x^p>0}  d^p,\nnb
\ee
i.e., the delay of the longest (slowest) path with positive rates from $s_i$ to $t_i$\footnote{We call a path $p\in P_i$ with $x^p>0$ as a flow-carrying path of $f_i$.}. The total delay of $f_i$ is defined as
\be
%\label{eqn:maxdelay}
\mathcal{T}(f_i)~~ \triangleq~~ \sum_{p\in P_i}  (x^p\cdot d^p)~~=~~\sum_{e\in E} (x_i^e\cdot d_e).\nnb
\ee
With $\mathcal{T}(f_i)$, we can easily define the \textbf{average delay} experienced by $f_i$ as $\mathcal{A}(f_i) \triangleq\mathcal{T}(f_i)/|f_i|$, and we let $\mathcal{A}(f_i)=0$ if $|f_i|=0$.

For each $f_i$, $i=1,2,...,K$, we denote its \textbf{throughput-based utility} as $\mathcal{U}_i^t(|f_i|)$, which is a function that rewards $f_i$ based on the achieved throughput. Similarly, we denote its \textbf{maximum-delay-based utility} as $-\mathcal{U}_i^d(\mathcal{M}(f_i))$, where $\mathcal{U}_i^d(\mathcal{M}(f_i))$ is a function that penalizes $f_i$ based on the experienced maximum delay. 

\subsection{Problem Definition}
In this paper, we study the following problem of Maximizing aggregate user Utilities subject to link capacity constraints, maximum Delay constraints, and Throughput requirements (\textbf{\textsf{MUDT}}),
\bse\label{eqn:MUDT}
\bee
(\textsf{MUDT}): \text{obj:} && \text{either }\max~ \sum_{i=1}^{K}\mathcal{U}_i^t(|f_i|), \label{eqn:OUDT-obj-1} \\
&& \text{or } \max~ -\sum_{i=1}^{K}\mathcal{U}_i^d(\mathcal{M}(f_i)),\label{eqn:OUDT-obj-2}\\
\text{s.t.} && |f_i|~\ge~ R_i,~~\forall i=1,2,...,K,\label{eqn:OUDT-throughput}\\
&& \mathcal{M}(f_i)~\le~D_i,~~\forall i=1,2,...,K,\label{eqn:OUDT-delay}\\
&& f=\{f_1,f_2,...,f_K\}\in \mathcal{X},\label{eqn:OUDT-feasible}
\eee
\ese
where $\mathcal{X}$ defines a feasible multiple-unicast flow $f$ meeting flow conservation constraints and link capacity constraints, i.e.,
\bee
\mathcal{X} & \triangleq \Bigg\{ \sum_{e\in \textsf{Out}(s_i)}x_i^e=\sum_{e\in \textsf{In}(t_i)}x_i^e = |f_i|,~\forall 1\le i\le K,\nnb\\
& \sum_{e\in \textsf{Out}(v)}x_i^e=\sum_{e\in \textsf{In}(v)}x_i^e,~\forall v\in V\backslash\{s_i,t_i\},~\forall 1\le i\le K,\nnb\\
& \sum_{i=1}^{K}x_i^e\le c_e,\forall e\in E,\text{~vars:~}x_i^e\ge 0,\forall e\in E,\forall 1\le i\le K\Bigg\}.\nnb
\eee

In formula~\eqref{eqn:MUDT}, the objective~\eqref{eqn:OUDT-obj-1} (resp.~\eqref{eqn:OUDT-obj-2}) maximizes the aggregate throughput-based utilities (resp. maximum-delay-based utilities) of all the users, the throughput requirements~\eqref{eqn:OUDT-throughput} require the throughput achieved by each user $i$ to be no smaller than $R_i$, the maximum delay constraints~\eqref{eqn:OUDT-delay} restrict the maximum delay experienced by each user $i$ to be no greater than $D_i$, and the feasibility constraint~\eqref{eqn:OUDT-feasible} defines a feasible multiple-unicast network flow solution, meeting link capacity constraints.

In the end of this section, we give an important theorem of \textsf{MUDT}, which argues that it is impossible even to (i) construct a feasible solution meeting all constraints, or (ii) obtain an optimal solution meeting relaxed constraints, in a polynomial time, unless P = NP. Thus it is non-trivial to develop polynomial-time approximation algorithms for \textsf{MUDT} subject to relaxed constraints.

\begin{theorem}\label{rmk:NP-hard}
	For \textsf{MUDT}, it is NP-complete (i) to construct a feasible solution that meets all constraints, or (ii) to obtain an optimal solution that meets throughput requirements but relaxes maximum delay constraints, or (iii) to obtain an optimal solution that meets maximum delay constraints but relaxes throughput requirements.
\end{theorem}
\begin{proof}
	Refer to our Appendix~\ref{adx:rmk-hard}.
	%Refer to Appendix 7.3 of our technical report~\cite{myreport}.
\end{proof}
\section{Proposed Algorithm \textsf{PASS}}\label{sec:framework}
In this section we design an algorithm \textsf{PASS} for \textsf{MUDT} of maximizing aggregate user utilities. We characterize conditions of the input utility functions such that \textsf{PASS} theoretically gives approximate solutions in a polynomial time, meeting relaxed constraints. 
%We emphasize again that \textsf{PASS} solves a class of maximum-delay-aware network flow problems in the multiple-unicast scenario, by leveraging a new understanding between non-convex maximum-delay-aware problems and convex average-delay-aware ones. \textsf{PASS} is in sharp contrast to the existing approximation technique~\cite{misra2009polynomial,zhang2010reliable}, which is limited to solving the specific maximum delay minimization problem in the single-unicast scenario, by solving problems iteratively in the time-expanded network using binary search.
\subsection{Algorithmic Structure of \textsf{PASS}}
We note that the non-convex maximum delays bring difficulties for solving \textsf{MUDT}. The key idea of our proposed \textsf{PASS} is to replace the non-convex maximum delays in \textsf{MUDT} by the convex average delays, and solve the average-delay-aware counterpart to obtain an approximate solution to \textsf{MUDT} in a polynomial time. (i) We denote the average-delay-aware counterpart of the \textsf{MUDT} that maximizes throughput-based utilities, i.e., problem~\eqref{eqn:MUDT} with an objective of~\eqref{eqn:OUDT-obj-1}, as \textbf{\textsf{MUAT-T}}, with the following formulation
\bse\label{eqn:OUAT-T}
\bee
(\textsf{MUAT-T}): \text{obj:} && \max~ \sum_{i=1}^{K}\mathcal{U}_i^t(|f_i|),  \\
\text{s.t.} && |f_i|~\ge~ R_i,~~\forall i=1,2,...,K,\label{eqn:OUAT-T-throughput}\\
&& \mathcal{T}(f_i)~\le~D_i\cdot |f_i|,~~\forall i=1,2,...,K,\label{eqn:OUAT-T-delay}\\
&& f=\{f_1,f_2,...,f_K\}\in \mathcal{X}.\label{eqn:OUAT-T-feasible}
\eee
\ese
(ii) Similarly, we denote the average-delay-aware counterpart of the \textsf{MUDT} that maximizes maximum-delay-based utilities, i.e., problem~\eqref{eqn:MUDT} with an objective of~\eqref{eqn:OUDT-obj-2}, as \textbf{\textsf{MUAT-M}}. \textsf{MUAT-M} has the following formulation
\bse\label{eqn:OUAT-D}
\bee
(\textsf{MUAT-M}): \text{obj:} && \max~ -\sum_{i=1}^{K}\mathcal{U}_i^d\left(\frac{\mathcal{T}(f_i)}{R_i}\right),  \\
\text{s.t.} && |f_i|~=~ R_i,~~\forall i=1,2,...,K,\label{eqn:OUAT-D-throughput}\\
&& \mathcal{T}(f_i)~\le~D_i\cdot R_i,~~\forall i=1,2,...,K,\label{eqn:OUAT-D-delay}\\
&& f=\{f_1,f_2,...,f_K\}\in \mathcal{X}.\label{eqn:OUAT-D-feasible}
\eee
\ese

Algorithm~\ref{algorithm:framework} describes the details of \textsf{PASS}. It first solves the average-delay-aware counterpart of the \textsf{MUDT} and obtain the corresponding multiple-unicast flow solution $f=\{f_i,i=1,2,...,K\}$ (line~\ref{line:original-solution}). Next for each $i=1,2,...,K$, we delete a rate of $\epsilon\cdot |f_i|$ iteratively from the slowest flow-carrying paths of $f_i$ (line~\ref{line:while}). In the end, the remaining flow is the solution returned by \textsf{PASS}. 
\subsection{\textsf{PASS} can Solve \textsf{MUDT} Approximately, Meeting Relaxed Constraints}\label{sec:condition}
Now we give an important lemma which will be used later to prove the approximation ratio of our \textsf{PASS}.  

\begin{algorithm} [t]
	\caption{Our Proposed Algorithm \textsf{PASS}}\label{algorithm:framework}
	\begin{algorithmic}[1]
		\State \textbf{input}: Problem~\eqref{eqn:MUDT}, $\epsilon \in (0,1)$
		\State \textbf{output}: $f=\{f_i,i=1,2,...,K\}$
		\Procedure{}{}
		\State \parbox[t]{\dimexpr\linewidth-\algorithmicindent}{Formulate either problem~\eqref{eqn:OUAT-T} or problem~\eqref{eqn:OUAT-D} that is the average-delay-aware counterpart of the input problem~\eqref{eqn:MUDT}}\label{line:hatP} 
		\State \parbox[t]{\dimexpr\linewidth-\algorithmicindent}{Solve the average-delay-aware problem and get the solution $f=\{f_i,i=1,2,...,K\}$}\label{line:original-solution}
		% \State //\texttt{we first get an edge-based flow and then do flow decomposition to get the path-based flow $f_{\textsf{SO}}(R)$, which has at most $|E|$ flow-carrying paths.}
		%	\State $P_{\textsf{SO}}$ = Flow-Decomposition($f_{\textsf{SO}}(R)$)  \label{line:do-flow-decomposition}
		%    \State \qquad \qquad // \texttt{path-based flow, at most $|E|$ paths}
		\State $x_i^{\textsf{delete}} = \epsilon\cdot |f_i|,\forall i=1,2,...,K$
		\For{$i=1,2,...,K$}
		\While{$x_i^{\textsf{delete}} > 0$}\label{line:while} 
		\State \parbox[t]{\dimexpr\linewidth-\algorithmicindent}{Find the slowest flow-carrying path $p_i\in P_i$ \strut}
		\If{$x^{p_i} > x_i^{\textsf{delete}}$} 
		\State $x^{p_i} = x^{p_i} - x_i^{\textsf{delete}},~x_i^{\textsf{delete}} = 0$
		\Else \label{line:delete-notenough}
		\State $x_i^{\textsf{delete}} = x_i^{\textsf{delete}}-x^{p_i},~x^{p_i} = 0$
		\EndIf
		\EndWhile
		\EndFor
		\State \textbf{return} the remaining flow $f=\{f_i,i=1,2,...,K\}$\label{line:final-solution}
		\EndProcedure
	\end{algorithmic}
\end{algorithm}

\begin{lemma}\label{lem:framework}
	In Algorithm~\ref{algorithm:framework} with an arbitrary $\epsilon\in(0,1)$, suppose $\hat{f}=\{\hat{f}_i,i=1,2,...,K\}$ is the solution to the average-delay-aware counterpart of \textsf{MUDT} (solution achieved in line~\ref{line:original-solution}), and suppose $\bar{f}=\{\bar{f}_i,i=1,2,...,K\}$ is the solution returned in the end (the remaining solution achieved in line~\ref{line:final-solution}). For any $i=1,2,...,K$, we have  
	%\be
	%\mathcal{M}[f_{\textsf{SO}^\textsf{d}}((1-\epsilon)R)]\leq \mathcal{M}[f_{\textsf{SO}}(R)],
	%\label{equ:maximum-dely-decrease-SO-D}
	%\ee
	%
	%\be
	%\mathcal{A}[f_{\textsf{DF}}((1-\epsilon)R)]\leq \mathcal{A}^*(R),
	%\label{equ:average-dely-decrease-SO-D}
	%\ee
	%
	\be
	\mathcal{T}\left(\bar{f}_i\right) + \epsilon\cdot \left|\hat{f}_i\right| \cdot \mathcal{M}\left(\bar{f}_i\right)
	\le \mathcal{T}\left(\hat{f}_i\right).
	\label{equ:total-dely-decrease-SO-D}
	\ee	
\end{lemma}
\begin{proof}
	Refer to our Appendix~\ref{adx:lem-frame}.
\end{proof}

Lem.~\ref{lem:framework} implies that $\epsilon\cdot\mathcal{M}(\bar{f}_i)\le \mathcal{A}(\hat{f}_i),\forall i=1,2,...,K$, i.e., the maximum delay of each single-unicast flow after deleting rate is bounded by a constant ratio as compared to the average delay of the corresponding single-unicast flow before deleting rate. With this critical observation that relates the non-convex maximum delays with the convex average delays, we can characterize conditions for \textsf{PASS} to solve \textsf{MUDT} approximately in a polynomial time.

\begin{theorem}\label{thm:sufficient-condition}
	Given a feasible problem~\eqref{eqn:MUDT}, suppose we use \textsf{PASS} (Algorithm~\ref{algorithm:framework}) with an arbitrary $\epsilon\in(0,1)$ to solve it. If the problem is feasible, meeting all conditions below
	\begin{enumerate}
%		\item $J$, $L$, and $M$ are all polynomial with $|E|$ and $K$,\label{ass:polynomial}
		\item for each $i=1,2,...,K$, for an arbitrary $a\ge 0$, $\mathcal{U}_i^t(a)$ is concave, non-decreasing, and non-negative with $a$, $\mathcal{U}_i^d(a)$ is convex, non-decreasing, and non-negative with $a$,\label{ass:convex-opt}
		%\item the following holds for $\mathcal{U}_t(\cdot)$ 
		%\be 
		%\mathcal{U}_t(a_1,a_2,...,a_K)=0,~~\text{if }a_1=a_2=...=a_K=0,\nnb
		%\ee 
		%\label{ass:utility-concave}
		\item \label{ass:penalty-convex} for an arbitrary $a\ge 0$, the following holds given any $\sigma\ge1$
		\be
		\mathcal{U}_i^d(\sigma\cdot a)~\le~ \sigma\cdot \mathcal{U}_i^d(a),~~\forall i=1,2,...,K,\nnb
		\ee 
	\end{enumerate}
	then \textsf{PASS} must return a solution $\bar{f}=\{\bar{f}_i,i=1,...,K\}$ in a polynomial time, meeting the following relaxed constraints
	\bse
	\label{eqn:condition-constraints}
	\bee
	&& \left|\bar{f}_i\right|~\ge~(1-\epsilon)\cdot R_i,~~\forall i=1,2,...,K,\label{eqn:relaxed-throughput} \\
	&& \mathcal{M}\left(\bar{f}_i\right)~\le~ D_i/\epsilon,~~\forall i=1,2,...,K, \label{eqn:relaxed-delay}\\
	&& \bar{f}=\{\bar{f}_1,\bar{f}_2,...,\bar{f}_K\}\in \mathcal{X}.\label{eqn:relaxed-feasible}
	\eee
	\ese
	Suppose $f^*=\{f_i^*,i=1,2,...,K\}$ is the optimal solution to the problem~\eqref{eqn:MUDT}. If the throughput-based utility maximization~\eqref{eqn:OUDT-obj-1} is the objective, $\bar{f}$ provides the following approximation ratio
	\be
	\sum_{i=1}^{K}\mathcal{U}_i^t\left(\left|\bar{f}_i\right|\right)~\ge~ (1-\epsilon)\cdot \sum_{i=1}^{K}\mathcal{U}_i^t\left(\left|f_i^*\right|\right).\\
	\ee
	If the maximum-delay-based utility maximization~\eqref{eqn:OUDT-obj-2} is the objective, $\bar{f}$ provides the following approximation ratio
	\be
	\sum_{i=1}^{K}\mathcal{U}_i^d\left(\mathcal{M}\left(\bar{f}_i\right)\right)~\le~ \frac{1}{\epsilon}\cdot \sum_{i=1}^{K}\mathcal{U}_i^d\left(\mathcal{M}\left(f_i^*\right)\right).\\
	\ee
\end{theorem}
\begin{proof}
Refer to our Appendix~\ref{adx:condition}.
\end{proof}

It is clear that \textsf{PASS} provides a constant approximation ratio, at the cost of violating throughput requirements~\eqref{eqn:OUDT-throughput} by a constant ratio of $(1-\epsilon)$, and
%In this case, \textsf{PASS} can support $(1-\epsilon)$-fraction of the required throughputs of individual unicast flows. This result can be of particular interest to delay-critical networking applications that are tolerant of a small traffic loss, e.g., video conferencing where a $3\%$ rate loss is acceptable with loss protection/recovery and error resilience capabilities~\cite{weinstein2008polycom}. 
violating maximum delay constraints~\eqref{eqn:OUDT-delay} by a constant ratio of $1/\epsilon$. For certain applications, the throughput requirements or the maximum delay constraints are hard constraints that cannot be violated. We note that one can use pre-scaled maximum delay constraints and throughput requirements as the input to \textsf{PASS} to generate feasible solutions as the output. Moreover, in the following, by slightly modifying \textsf{PASS}, we respectively develop (i) an algorithm \textsf{PASS-M} to achieve approximate solutions that can strictly meet maximum delay constraints, and (ii) an algorithm \textsf{PASS-T} to achieve approximate solutions that can strictly meet throughput requirements.

\subsection{Modify \textsf{PASS} to Strictly Meet Maximum Delay Constraints}
\begin{algorithm} [t]
	\caption{\textsf{PASS-M}: Modify \textsf{PASS} to Strictly Meet Maximum Delay Constraints }\label{alg:PASS-M}
	\begin{algorithmic}[1]
		\State \textbf{input}: Problem~\eqref{eqn:MUDT}
		\State \textbf{output}: $f=\{f_i,i=1,2,...,K\}$
		\Procedure{}{}
		\State \parbox[t]{\dimexpr\linewidth-\algorithmicindent}{Solve the average-delay-aware counterpart of the problem~\eqref{eqn:MUDT}, and get the solution $f=\{f_i,i=1,2,...,K\}$}\label{line:PASS-M-original-solution}
		\For{$i=1,2,...,K$}
		\While{$\mathcal{M}(f_i)>D_i$}
		\State Find the slowest flow-carrying path $p_i\in P_i$
		\State Let $x^{p_i}=0$
		\EndWhile
		\EndFor
		% \State //\texttt{we first get an edge-based flow and then do flow decomposition to get the path-based flow $f_{\textsf{SO}}(R)$, which has at most $|E|$ flow-carrying paths.}
		%	\State $P_{\textsf{SO}}$ = Flow-Decomposition($f_{\textsf{SO}}(R)$)  \label{line:do-flow-decomposition}
		%    \State \qquad \qquad // \texttt{path-based flow, at most $|E|$ paths}
		%\State $x_i^{\textsf{delete}} = 0,\forall i=1,2,...,K$
		%\For{$i=1,2,...,K$}
		%\While{$\mathcal{M}(f_i)>M_i$} 
		%\State \parbox[t]{\dimexpr\linewidth-\algorithmicindent}{Find the slowest flow-carrying path $p_l\in P_i$ \strut}
		%\State $x_i^{\textsf{delete}} = x_i^{\textsf{delete}}+x^{p_l}$
		%\State $x^{p_l} = 0$
		%\EndWhile
		%\EndFor
		\State \textbf{return} the remaining flow $f=\{f_i,i=1,2,...,K\}$
		\EndProcedure
	\end{algorithmic}
\end{algorithm}

We introduce \textsf{PASS-M} in Algorithm~\ref{alg:PASS-M}. Similar to \textsf{PASS}, \textsf{PASS-M} first solves the average-delay-aware counterpart of \textsf{MUDT}. But different from \textsf{PASS} that deletes $\epsilon\cdot|f_i|$ rate from slowest flow-carrying paths of each $f_i$, \textsf{PASS-M} deletes rate from slowest flow-carrying paths of $f_i$ till the maximum delay of $f_i$ strictly meets the constraint $D_i$. In the following theorem, we prove that \textsf{PASS-M} can obtain a solution with a problem-dependent approximation ratio.
\begin{theorem}\label{thm:PASS-M}
	Given a feasible problem~\eqref{eqn:MUDT}, suppose it meets all conditions in Thm.~\ref{thm:sufficient-condition}. Suppose we use \textsf{PASS-M} (Algorithm~\ref{alg:PASS-M}) to solve it. Then \textsf{PASS-M} must return a solution $\bar{f}=\{\bar{f}_i,i=1,2,...,K\}$ in a polynomial time, meeting the following relaxed constraints
	\bse
	\bee
	&& \left|\bar{f}_i\right|~\ge~(1-\epsilon_{\max})\cdot R_i,~~\forall i=1,2,...,K,\label{eqn:PASS-M-throughput}  \\
	&& \mathcal{M}\left(\bar{f}_i\right)~\le~ D_i,~~\forall i=1,2,...,K,\label{eqn:PASS-M-delay} \\
	&& \bar{f}=\{\bar{f}_1,\bar{f}_2,...,\bar{f}_K\}\in \mathcal{X},\label{eqn:PASS-M-feasible}
	\eee
	\ese
	where $\epsilon_{\max}$ is defined as follows
	\be
	\epsilon_{\max}~=~\max_{1\le i\le K}\left\{\left(\left|\hat{f}_i\right|-\left|\bar{f}_i\right|\right)/\left|\hat{f}_i\right|\right\},\nnb
	\ee
	where $\hat{f}=\{\hat{f}_i,i=1,2,...,K\}$ is the optimal solution to the average-delay-aware problem in line~\ref{line:PASS-M-original-solution} of Algorithm~\ref{alg:PASS-M}. Suppose $f^*=\{f_i^*,i=1,2,...,K\}$ is the optimal solution to problem~\eqref{eqn:MUDT}. If the throughput-based utility maximization~\eqref{eqn:OUDT-obj-1} is the objective, $\bar{f}$ provides the following approximation ratio
	\be
	\sum_{i=1}^{K}\mathcal{U}_i^t\left(\left|\bar{f}_i\right|\right)~\ge~ (1-\epsilon_{\max})\cdot \sum_{i=1}^{K}\mathcal{U}_i^t\left(\left|f_i^*\right|\right).\label{eqn:PASS-M-T-ratio}\\
	\ee
	If the maximum-delay-based utility maximization~\eqref{eqn:OUDT-obj-2} is the objective, $\bar{f}$ provides the following approximation ratio
	\be~
	\sum_{i=1}^{K}\mathcal{U}_i^d\left(\mathcal{M}\left(\bar{f}_i\right)\right)~\le~ \frac{1}{\epsilon_{\min}}\cdot \sum_{i=1}^{K}\mathcal{U}_i^d\left(\mathcal{M}\left(f_i^*\right)\right),\label{eqn:PASS-M-D-ratio}\\
	\ee
    where $\epsilon_{\min}$ is defined as follows
	\be
	\epsilon_{\min}~=~\min_{1\le i\le K}\left\{\left(\left|\hat{f}_i\right|-\left|\bar{f}_i\right|\right)/\left|\hat{f}_i\right|\right\}.\nnb
	\ee
\end{theorem}
\begin{proof}
	Refer to our Appendix~\ref{adx:PASS-M}.
	%Refer to Appendix 7.4 of our technical report~\cite{myreport}.
\end{proof}

Comparing Thm.~\ref{thm:sufficient-condition} of \textsf{PASS} with Thm.~\ref{thm:PASS-M} of \textsf{PASS-M}, to solve \textsf{MUDT}, (i) \textsf{PASS} achieves a solution with a constant approximation ratio, at the cost of violating both throughput requirements and maximum delay constraints by constant ratios, while (ii) \textsf{PASS-M} obtains a solution with a problem-dependent approximation ratio, strictly meeting maximum delay constraints, but at the cost of violating throughput requirements by a problem-dependent ratio.
%In the following section, we observe that many critical but NP-hard delay-aware routing problems in the literature satisfy our conditions, and hence can be tackled by \textsf{PASS} quickly and efficiently.
\subsection{Modify \textsf{PASS} to Strictly Meet Throughput Requirements}
In order to strictly meet throughput requirements, our \textsf{PASS-T} suggest to use the optimal solution to the average-delay-aware counterpart of \textsf{MUDT} directly as a solution to the maximum-delay-aware problem \textsf{MUDT}, i.e.,

$\mbox{\ensuremath{\rhd}}$ \textsf{PASS-T}: directly solve the average-delay-aware counterpart of the problem~\eqref{eqn:MUDT}.

\begin{theorem}\label{thm:PASS-T}
	Given a feasible problem~\eqref{eqn:MUDT}, suppose it meets all conditions in Thm.~\ref{thm:sufficient-condition}. We denote $\bar{g}=\{\bar{g}_1,\bar{g}_2,...,\bar{g}_K\}$ as the solution returned if we use \textsf{PASS} (Algorithm~\ref{algorithm:framework}) to solve it with an $\epsilon\in(0,1)$. Now suppose we use \textsf{PASS-T} to solve the problem~\eqref{eqn:MUDT}. Then \textsf{PASS-T} must return a solution $\bar{f}=\{\bar{f}_i,i=1,2,...,K\}$ in a polynomial time, meeting the following relaxed constraints
	\bse
	\bee
	&& \left|\bar{f}_i\right|~\ge~R_i,~~\forall i=1,2,...,K,\label{eqn:PASS-T-throughput} \\
	&& \mathcal{M}\left(\bar{f}_i\right)~\le~ \frac{\lambda}{\epsilon}\cdot D_i,~~\forall i=1,2,...,K,\label{eqn:PASS-T-delay}\\
	&& \bar{f}=\{\bar{f}_1,\bar{f}_2,...,\bar{f}_K\}\in \mathcal{X},\label{eqn:PASS-T-feasible}
	\eee
	\ese
	where $\lambda$ is defined as follows
	\be
	\lambda~=~\max\left\{1,\max_{1\le i\le K}\left\{\mathcal{M}(\bar{f}_i)/\mathcal{M}(\bar{g}_i)\right\}\right\}.\nnb
	\ee
	Suppose $f^*=\{f_i^*,i=1,2,...,K\}$ is the optimal solution to problem~\eqref{eqn:MUDT}. If the throughput-based utility maximization~\eqref{eqn:OUDT-obj-1} is the objective, $\bar{f}$ provides the following approximation ratio
	\be
	\sum_{i=1}^{K}\mathcal{U}_i^t\left(\left|\bar{f}_i\right|\right)~\ge~ \sum_{i=1}^{K}\mathcal{U}_i^t\left(\left|f_i^*\right|\right).\label{eqn:PASS-T-T-ratio}\\
	\ee
	If the maximum-delay-based utility maximization~\eqref{eqn:OUDT-obj-2} is the objective, $\bar{f}$ provides the following approximation ratio
	\be
	\sum_{i=1}^{K}\mathcal{U}_i^d\left(\mathcal{M}\left(\bar{f}_i\right)\right)~\le~ \frac{\lambda}{\epsilon}\cdot \sum_{i=1}^{K}\mathcal{U}_i^d\left(\mathcal{M}\left(f_i^*\right)\right).\label{eqn:PASS-T-D-ratio}\\
	\ee
\end{theorem}
\begin{proof}
	Refer to our Appendix~\ref{adx:PASS-T}.
	%Refer to Appendix 7.5 of our technical report~\cite{myreport}.
\end{proof}

Thm.~\ref{thm:PASS-T} suggests that we can figure out an approximation ratio of \textsf{PASS-T} with the knowledge of an arbitrary solution of \textsf{PASS}. Comparing Thm.~\ref{thm:sufficient-condition} of \textsf{PASS} with Thm.~\ref{thm:PASS-T} of \textsf{PASS-T}, in order to solve \textsf{MUDT}, (i) \textsf{PASS} achieves a solution with a constant approximation ratio, at the cost of violating both throughput requirements and maximum delay constraints by constant ratios, while (ii) \textsf{PASS-T} obtains a solution with a problem-dependent approximation ratio, strictly meeting throughput requirements, but at the cost of violating maximum delay constraints by a problem-dependent ratio.

\subsection{Our Proposed Algorithms Can Solve Other Maximum-Delay-Aware Problems}\label{subsec:other-problem}
As shown in problem~\eqref{eqn:MUDT}, \textsf{MUDT} has an objective of either~\eqref{eqn:OUDT-obj-1} or ~\eqref{eqn:OUDT-obj-2}, both of which maximize aggregate user utilities. Differently, another two representative user-utility-sensitive objectives are
\bse
\bee
&&\max \min_{1\le i\le K}\left\{\mathcal{U}_i^t(|f_i|)\right\},\label{eqn:new-obj-1}\\
&&\max \min_{1\le i\le K}\left\{-\mathcal{U}_i^d(\mathcal{M}(f_i))\right\},\label{eqn:new-obj-2} 
\eee
\ese
both of which maximize worst user utilities. Following same proof to Thm.~\ref{thm:sufficient-condition}, Thm.~\ref{thm:PASS-M}, and Thm.~\ref{thm:PASS-T}, it is easy to verify that as long as the conditions in Thm.~\ref{thm:sufficient-condition} are satisfied, we can use \textsf{PASS}, \textsf{PASS-M}, and \textsf{PASS-T} to solve the problem with an objective of either~\eqref{eqn:new-obj-1} or~\eqref{eqn:new-obj-2}, subject to throughput requirements~\eqref{eqn:OUDT-throughput}, maximum delay constraints~\eqref{eqn:OUDT-delay}, and feasibility constraints~\eqref{eqn:OUDT-feasible}, approximately in a polynomial time. Our design of \textsf{PASS} suggests a new avenue for solving maximum-delay-aware network optimization problems. 

Overall in this section, we design \textsf{PASS} to solve the maximum-delay-aware problem \textsf{MUDT} approximately in a polynomial time under practical conditions. \textsf{PASS} solves the average-delay-aware counterpart of \textsf{MUDT} only once in the input network, and then deletes certain flow rate from slowest flow-carrying paths to obtain solutions with theoretical performance guarantee. Note again that in sharp contrast, existing maximum-delay-aware problems either minimize throughput-constrained maximum delay or maximize maximum-delay-constrained throughput, which are special cases of our problem \textsf{MUDT}. They rely on a time-consuming technique of solving problems iteratively in the time-expanded network to provide approximate solutions. Our \textsf{PASS} leverages a novel understanding between non-convex maximum-delay-aware problems and their convex average-delay-aware counterparts, which can be of independent interest and suggest a new avenue for solving maximum-delay-aware network optimization problems.
\section{Popular Delay-/Throughput- Aware Network Communication Scenarios}\label{sec:example}
In this section we introduce several popular network communication settings that are sensitive both to the throughputs and to the maximum delays. Although associated problems are all NP-hard, we observe that they are all special cases of \textsf{MUDT}, and all satisfy conditions introduced in Thm.~\ref{thm:sufficient-condition}, and hence can be solved by \textsf{PASS}, \textsf{PASS-M}, and \textsf{PASS-T} approximately with strong theoretical performance guarantee in a polynomial time.
\subsection{Throughput-Constrained Maximum Delay Minimization}
The Throughput-Constrained maximum Delay Minimization problem (\textsf{TCDM}) aims to find a network flow to minimize the weighted summation of maximum delays of all users, subject to link capacity constraints and throughput requirements.
\bse
\label{eqn:TCTM}
\bee
(\textsf{TCDM}): \min && \sum_{i=1}^{K} \left(w_i\cdot \mathcal{M}(f_i)\right) \label{eqn:TCTMobj}
\\ \mbox{s.t. } && |f_i| \ge R_i, ~~\forall i=1,2,...,K,\label{eqn:TCTMrate}\\
&& f=\{f_1,f_2,...,f_K\}\in \mathcal{X},\label{eqn:TCTMfeasible}
\eee
\ese
where in the objective~\eqref{eqn:TCTMobj} a non-negative weight $w_i\ge 0$ is associated with the maximum delay of $f_i$ for each $i=1,2,...,K$.

\textsf{TCDM} is NP-hard, since as its special case when $K=1$, the single-unicast maximum delay minimization problem is known to be NP-hard~\cite{misra2009polynomial}. Maximum delay minimization problems similar to \textsf{TCDM} have been studied in~\cite{misra2009polynomial,zhang2010reliable,correa2004computational,correa2007fast,my}. It is clear that \textsf{TCDM} satisfies our conditions introduced in Thm.~\ref{thm:sufficient-condition}. Therefore, by replacing the non-convex maximum delays with the convex average delays, we can get the average-delay-aware counterpart formulated in the way of problem~\eqref{eqn:OUAT-D}, and thus can either (i) use \textsf{PASS} to solve \textsf{TCDM} with a constant approximation ratio while violating throughput requirements also by a constant ratio (see Thm.~\ref{thm:sufficient-condition}), or (ii) use \textsf{PASS-T} to solve \textsf{TCDM} with a problem-dependent approximation ratio, strictly meeting throughput requirements (see Thm.~\ref{thm:PASS-T}). 

%Different from the objective of minimizing the weighted summation of maximum delays~\eqref{eqn:TCTMobj}, another representative maximum-delay-aware objective is to minimize the worst maximum delay performances of all the users, namely
%\be
%\label{eqn:obj-min-maxd}
%\min \max_{1\le i\le K}\left\{\mathcal{M}(f_i)\right\}.
%\ee
%Although the objective~\eqref{eqn:obj-min-maxd} is not to optimize aggregate user utilities, following the same proof to Thm.~\ref{thm:sufficient-condition} and Thm.~\ref{thm:PASS-T}, it is easy to verify that the maximum delay minimization problem with an objective of~\eqref{eqn:obj-min-maxd} subject to constraints~\eqref{eqn:TCTMrate} and~\eqref{eqn:TCTMfeasible} can be solved by \textsf{PASS} and \textsf{PASS-T} approximately in a polynomial time, too.
  
\subsection{Maximum-Delay-Constrained Throughput-Based Utility Maximization}\label{subsec:theoreical-DCUM}
The maximum-Delay-Constrained throughput-based Utility Maximization (\textsf{DCUM}) problem aims to find a network flow to maximize aggregate user utilities, subject to link capacity constraints and maximum delay constraints. It has the following formulation.
\bse
\label{eqn:1Problem}
\bee
(\textsf{DCUM}): \max && \sum_{i=1}^{K} \mathcal{U}_i^t\left(|f_i|\right) \label{eqn:1obj}
\\ \mbox{s.t. }
&& \mathcal{M}(f_i)\le D_i,~~\forall i=1,2,...,K,\label{eqn:1delay}\\
&& f=\{f_1,f_2,...,f_K\}\in \mathcal{X}.
\eee
\ese

\textsf{DCUM} is NP-hard, because as its special case when $K=1$ and $\mathcal{U}_1^t|f_1|=|f_1|$, the problem can be proved to be NP-hard following a similar proof as introduced in the Appendix of~\cite{misra2009polynomial}. Throughput-based utility maximization problems similar to \textsf{DCUM} have been studied in~\cite{cao2017optimizing,yu2018application}. Due to practical concerns, it is fair to assume that the throughput-based utility function of each user is concave, non-decreasing, and non-negative with the achieved throughput, thus meeting conditions introduced in our Thm.~\ref{thm:sufficient-condition}. After replacing the non-convex maximum delays with the convex average delays, we can get the average-delay-aware counterpart formulated in the way of problem~\eqref{eqn:OUAT-T}, and thus can either (i) use \textsf{PASS} to solve \textsf{DCUM} with a constant approximation ratio while violating maximum delay constraints also by a constant ratio (see Thm.~\ref{thm:sufficient-condition}), or (ii) use \textsf{PASS-M} to solve \textsf{DCUM} with a problem-dependent approximation ratio, strictly meeting maximum delay constraints (see Thm.~\ref{thm:PASS-M}). 

%In the problem~\eqref{eqn:1Problem}, the objective~\eqref{eqn:1obj} maximizes the aggregate user utilities. If it is replaced by the following objective of maximizing the minimum utility of all the users,
%\be
%\max \min_{1\le i\le K}\left\{\mathcal{U}_i^t(|f_i|)\right\},\label{eqn:obj-max-minU}
%\ee
%following the same proof to Thm.~\ref{thm:sufficient-condition} and Thm.~\ref{thm:PASS-M}, it can be proved that the resulting problem can be solved by \textsf{PASS} and \textsf{PASS-M} approximately in a polynomial time, too, despite that the objective~\eqref{eqn:obj-max-minU} is not to maximize aggregate user utilities.
%\input{pseudo.tex}
%\input{appro.tex}
%\input{VaryT.tex}
\section{Performance Evaluation}\label{sec:experiments}
\begin{figure}[]
	\centering
	\includegraphics[width=0.9\columnwidth]{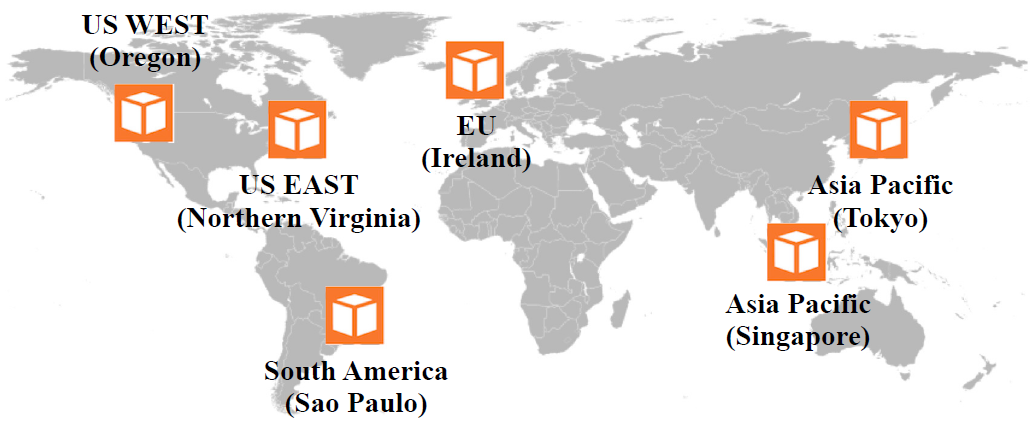}
	\caption{Topology of the 6 Amazon EC2 datacenters~\cite{liu2016delay}.}
	\label{fig:Amazon}
\end{figure}

We evaluate the empirical performance of our proposed algorithms, by simulating the delay-critical video conferencing traffic over a real-world continent-scale inter-datacenter network topology of 6 globally distributed Amazon EC2 datacenters (see Fig.~\ref{fig:Amazon}). The network is modeled as a complete undirected graph. Each undirected link is treated as two directed links that operate independently and have identical delays and capacities, a common way to model an undirected graph by a directed one, e.g. in~\cite{grimmer2016nash}. We set link delays and capacities according to practical evaluations on Amazon EC2 from~\cite{hajiesmaili2017cost,liu2016delay} (see Tab.~\ref{tab:delay}). We assume two unicasts, namely $K=2$, with $s_1$ to be Virginia, $t_1$ to be Singapore, $s_2$ to be Oregon, and $t_2$ to be Tokyo. Our test environment is an Intel Core i5 (2.40 GHz) processor with 8 GB memory running Windows 64-bit operating system. All the experiments are implemented in C++ and linear programs are solved using \emph{CPLEX}~\cite{cplex}.
\begin{table}[]
	\centering
	\caption{Information of $(d_e,c_e)$ for each link $e\in E$ in the Amazon EC2 network~\cite{hajiesmaili2017cost,liu2016delay}, where $d_e$ is link delay (in ms) and $c_e$ is link capacity (in Mbps), (OR: Oregon, VA: Virginia, IR: Ireland, TO: Tokyo, SI: Singapore, SP: Sao Paulo).}
	\label{tab:delay}
	\scalebox{0.9}{\begin{tabular}{|c *{6}{|c}|}
			%{|c|c|c|c|c|c|c|}
			\hline
			& OR & VA & IR & TO & SI & SP \\ \hline
			OR & N/A & (41,82) & (86,86) & (68,138) & (117,74) & (104,67) \\ \hline
			VA & - & N/A & (54,72) & (101,41) & (127,52) & (82,70) \\ \hline
			IR & - & - & N/A & (138,56) & (117,44) & (120,61) \\ \hline
			TO & - & - & - & N/A & (45,166) & (151,41) \\ \hline
			SI & - & - & - & - & N/A & (182,33) \\ \hline
			SP & - & - & - & - & - & N/A \\ \hline		
		\end{tabular}}
	\end{table}
\subsection{Use \textsf{PASS} to Minimize Maximum Delay}\label{subsec:TCDM}
\begin{figure}[!t]
	%\centering
	%\vspace{-3ex}
	\subfigure[Delay results with $\epsilon$ of \textsf{PASS}, with $R_1=R_2=230$. ]{\includegraphics[width=1.49in]{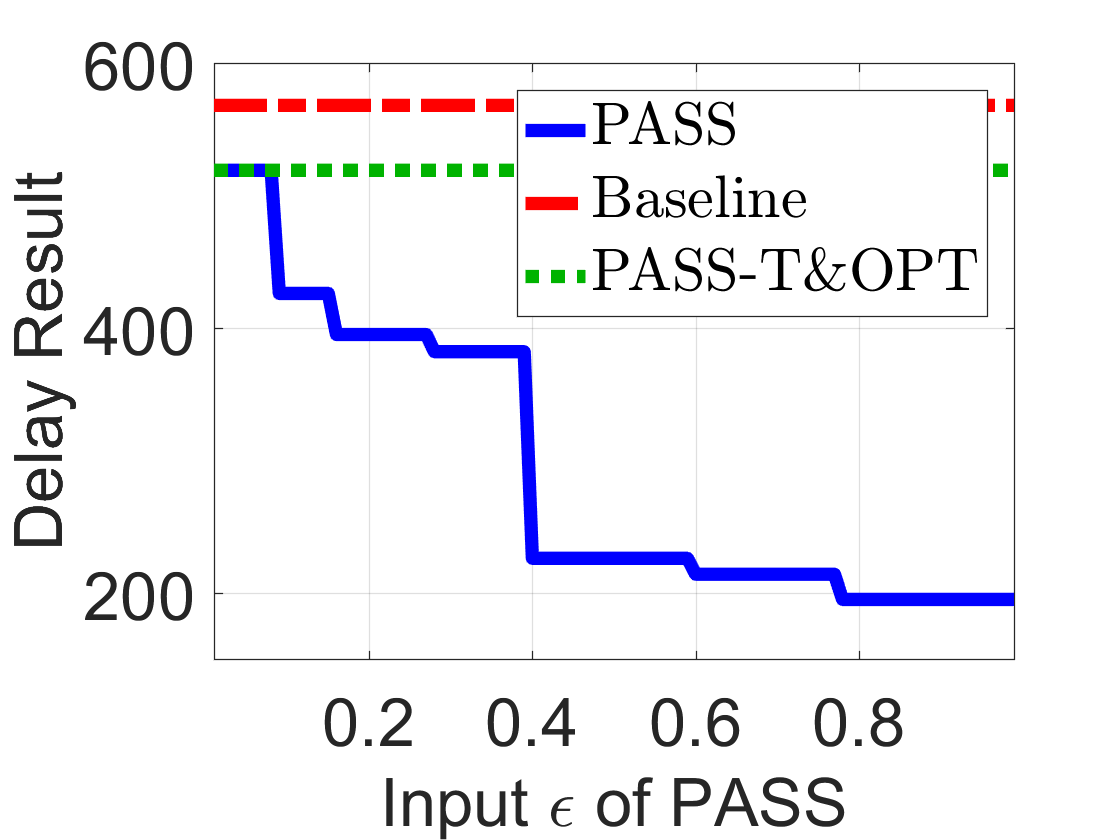}
		\label{subfig:TCDM-epsilon}}
	%\hfil
	\subfigure[Delay results with throughput requirements, with $\epsilon=3\%$ in \textsf{PASS}.]{\includegraphics[width=1.49in]{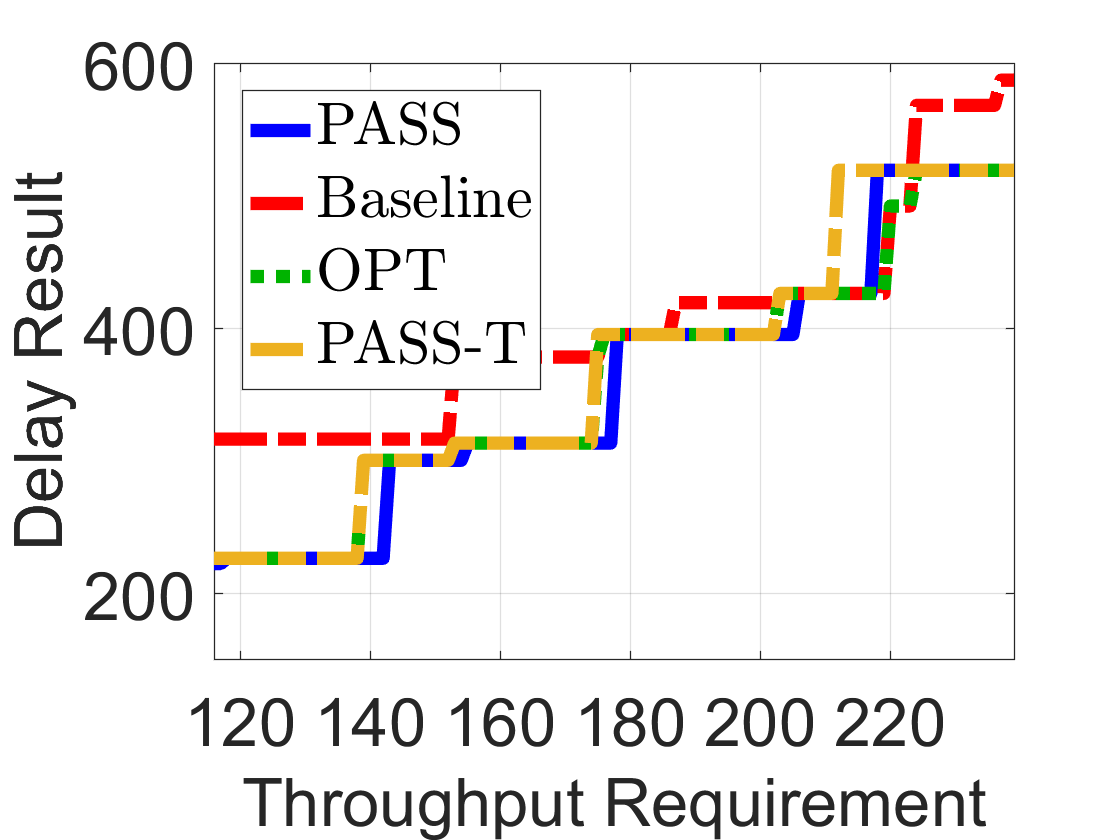}
		\label{subfig:TCDM-rate}}
	\caption{Simulation results of using \textsf{PASS} to minimize the summation of maximum delays.}
	\label{fig:TCDM}
\end{figure}
We now use \textsf{PASS} to minimize the maximum delay, subject to link capacity constraints and throughput requirements (i.e., to solve \textsf{TCDM} with formula~\eqref{eqn:TCTM}). We assume $w_1=w_2$ and $R_1=R_2=R$ in the formula~\eqref{eqn:TCTM}.

We compare \textsf{PASS} with the optimal solution, a conceivable greedy baseline, and \textsf{PASS-T} respectively. (i) Because link delays are all integers (see Tab.~\ref{tab:delay}), the delay of any path must be an integer. Therefore, we can obtain the optimal solution minimizing the summation of maximum delays, by enumerating all possible maximum delays of individual unicasts to figure out the minimal performance such that a feasible flow exists in the time-expanded network. Note that this approach theoretically has an exponential time complexity, and is the foundation of the \textsf{FPTAS}~\cite{misra2009polynomial} designed for the single-unicast maximum delay minimization problem. (ii) In order to minimize delay while satisfying throughput requirements, the baseline greedily obtains the routing solution from the unicast $1$ to the unicast $K$ one by one. In the iteration of the unicast $i$, it assigns as much rate as possible to the shortest paths from $s_i$ to $t_i$ iteratively respecting the link capacity constraints, till the throughput requirement $R_i$ is satisfied. Similar heuristic approaches have been used in other delay-aware network flow studies, e.g., in~\cite{devetak2011minimizing}. 

First, we evaluate the summation of maximum delays of \textsf{PASS} with $\epsilon$ (see Fig.~\ref{subfig:TCDM-epsilon}). We set $R=230$ and vary $\epsilon$ from $1\%$ to $99\%$ by a step of $1\%$. According to the figure, (i) \textsf{PASS-T} obtains the optimal solution to our problem, (ii) the delay of the baseline is strictly larger than optimal, and (iii) the delay of \textsf{PASS} is a staircase function with $\epsilon$. We remark that the delay of \textsf{PASS} can be smaller than optimal in many instances because \textsf{PASS} can only support $(1-\epsilon)$-fraction of the throughput requirement, while the optimal solution achieves the minimal summation of maximum delays among network flows supporting the full throughput requirement.

Second, we evaluate the summation of maximum delays of \textsf{PASS} with the throughput requirement $R$ (see Fig.~\ref{subfig:TCDM-rate}). We set $\epsilon=3\%$ since a $3\%$ throughput loss is very acceptable for video conferencing with protection/recovery capabilities~\cite{weinstein2008polycom}. We vary $R$ from $116$ to $239$ with a unit step. We remark that $116$Mbps is the smallest throughput when the baseline needs multiple paths to forward it for each of the two unicasts, and $239$Mbps is the largest throughput that can be routed. From Fig.~\ref{subfig:TCDM-rate}, it is clear that \textsf{PASS} outputs a smaller maximum delay compared with the baseline in most instances. In average, the maximum delay of the baseline ($402$) is over $11\%$ more than that of the optimal ($362$) and of the \textsf{PASS} ($359$). In the worst case ($R\in[116,138]$), the maximum delay of the baseline is over $40\%$ more than that of the optimal and of the \textsf{PASS}. In addition, \textsf{PASS-T} obtains the optimal solution to our problem in most instances, except for instances where $R\in[212,223]$. 
\subsection{Use \textsf{PASS} to Maximize Throughput}\label{subsec:DCUM}
\begin{figure}[!t]
	%\centering
	%\vspace{-3ex}
	\subfigure[Throughput results (both baseline and \textsf{PASS-M} obtain the optimal). ]{\includegraphics[width=1.49in]{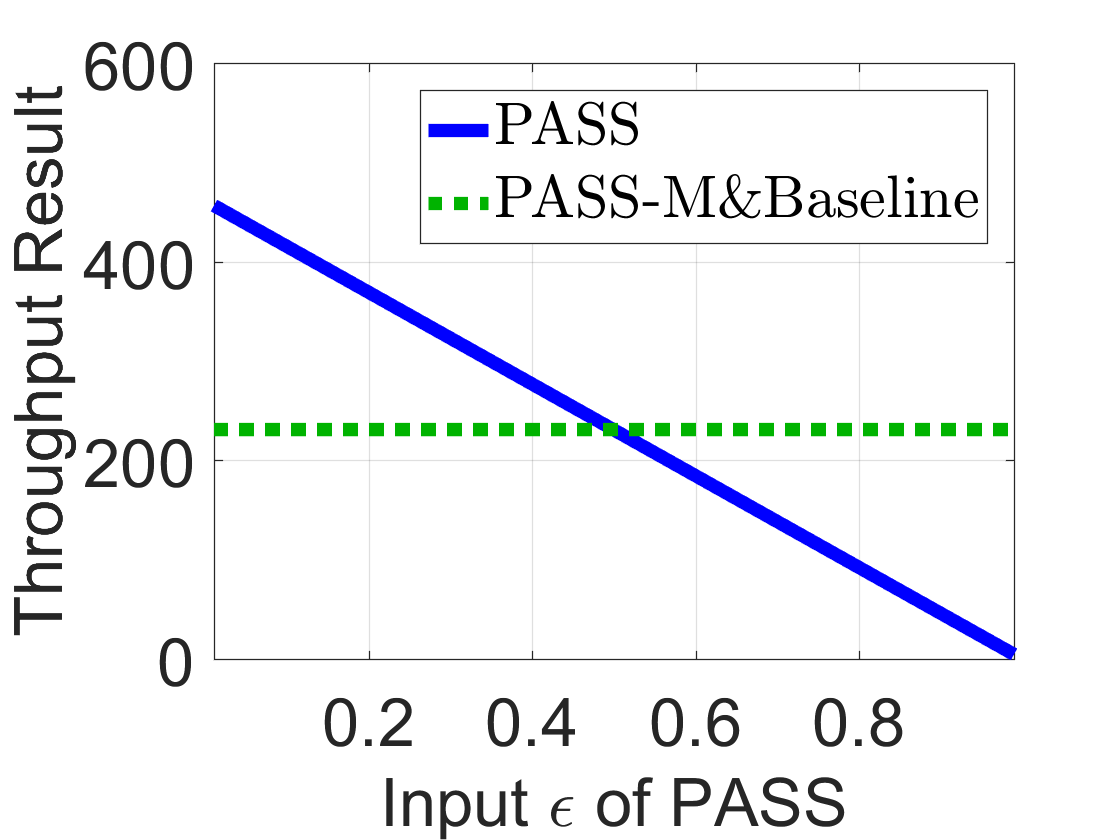}
		\label{subfig:Rate-DCUM}}
	%\hfil
	\subfigure[Delay ratio comparing the achieved result to the constraint.]{\includegraphics[width=1.49in]{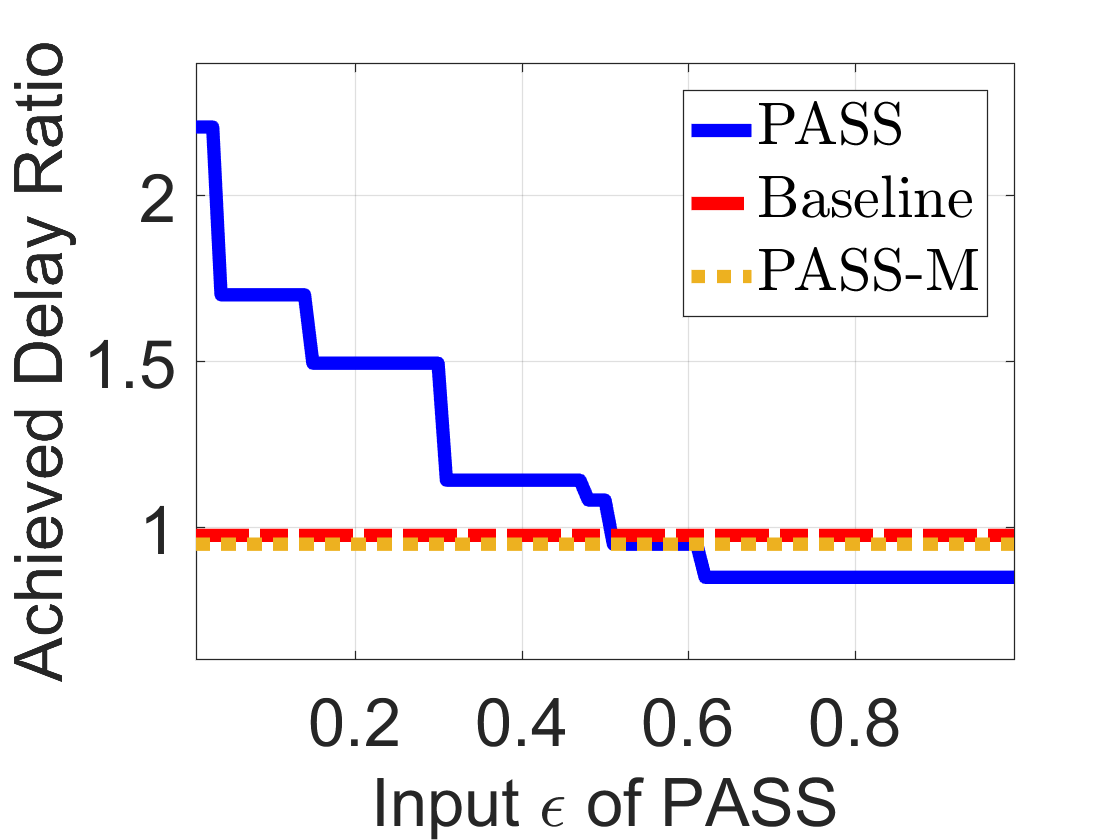}
		\label{subfig:Delay-DCUM}}
	\caption{Simulation results of using \textsf{PASS} to maximize total throughput with various $\epsilon$, where $D_1=D_2=150$.}
	\label{fig:DCUM}
\end{figure}
We then use \textsf{PASS} to maximize the throughput, subject to link capacity constraints and maximum delay constraints (i.e., to solve \textsf{DCUM} with formula~\eqref{eqn:1Problem}). We assume $\mathcal{U}_1^t(|f_1|)=|f_1|$, $\mathcal{U}_2^t(|f_2|)=|f_2|$, and $D_1=D_2=D$ in the formula~\eqref{eqn:1Problem}. We compare \textsf{PASS} with the optimal solution, a conceivable baseline, and \textsf{PASS-M}, respectively. Similar to the greedy approach introduced in Sec.~\ref{subsec:TCDM}, the baseline assigns as much rate as possible to the shortest paths respecting both link capacity constraints and maximum delay constraints iteratively from the unicast $1$ to the unicast $K$ one by one. Besides, similar to Sec.~\ref{subsec:TCDM}, we can obtain the optimal solution maximizing throughput by solving multiple-unicast flow problems in the time-expanded network. 

We set $D=150$ due to the following two concerns. (i) An end-to-end delay less than $150$ms can provide a transparent interactivity for video conferencing~\cite{ITU}. (ii) A delay larger than $150$ms (as long as it is less than $400$ms) is still acceptable for video conferencing~\cite{ITU}, and hence a solution that violates the maximum delay constraint (e.g., the solution of \textsf{PASS}) may still be useful if it can achieve a huge amount of throughput increment.

We vary $\epsilon$ from $1\%$ to $99\%$ with a step of $1\%$. We give the throughput results in Fig.~\ref{subfig:Rate-DCUM}, and the achieved maximum delay ratio results, i.e., $\max\{\mathcal{M}(f_1),\mathcal{M}(f_2)\}/D$ where $f$ is the solution, in Fig.~\ref{subfig:Delay-DCUM}. In our simulations, both the baseline and \textsf{PASS-M} obtain the optimal throughput strictly meeting maximum delay constraints. For $\epsilon\le 49\%$, the throughput of \textsf{PASS} is strictly larger than the optimal, while violating maximum delay constraints (e.g., $8\%$ more than $D$ when $\epsilon=49\%$). For $\epsilon\ge 51\%$, the solution of \textsf{PASS} meets maximum delay constraints, but the achieved throughput is strictly smaller than optimal. It is impressive that with a small $\epsilon$, e.g., $\epsilon=1\%$, the throughput of \textsf{PASS} is over $90\%$ more than optimal, while in the same time the maximum delays of \textsf{PASS} are less than $331$ms which is still acceptable for video conferencing. In average, we observe a $2.0\%$ throughput increment as compared to optimal, but with a $2.2\%$ violation with the maximum delay constraints, when $\epsilon$ is decreased by $1\%$ for instances where $\epsilon\le 49\%$.   

\subsection{Use \textsf{PASS} to Maximize Network Utility}
\begin{figure}[!t]
	%\centering
	%\vspace{-3ex}
	\subfigure[Network utility results of different algorithms, with $\epsilon=3\%$ in \textsf{PASS}. ]{\includegraphics[width=1.49in]{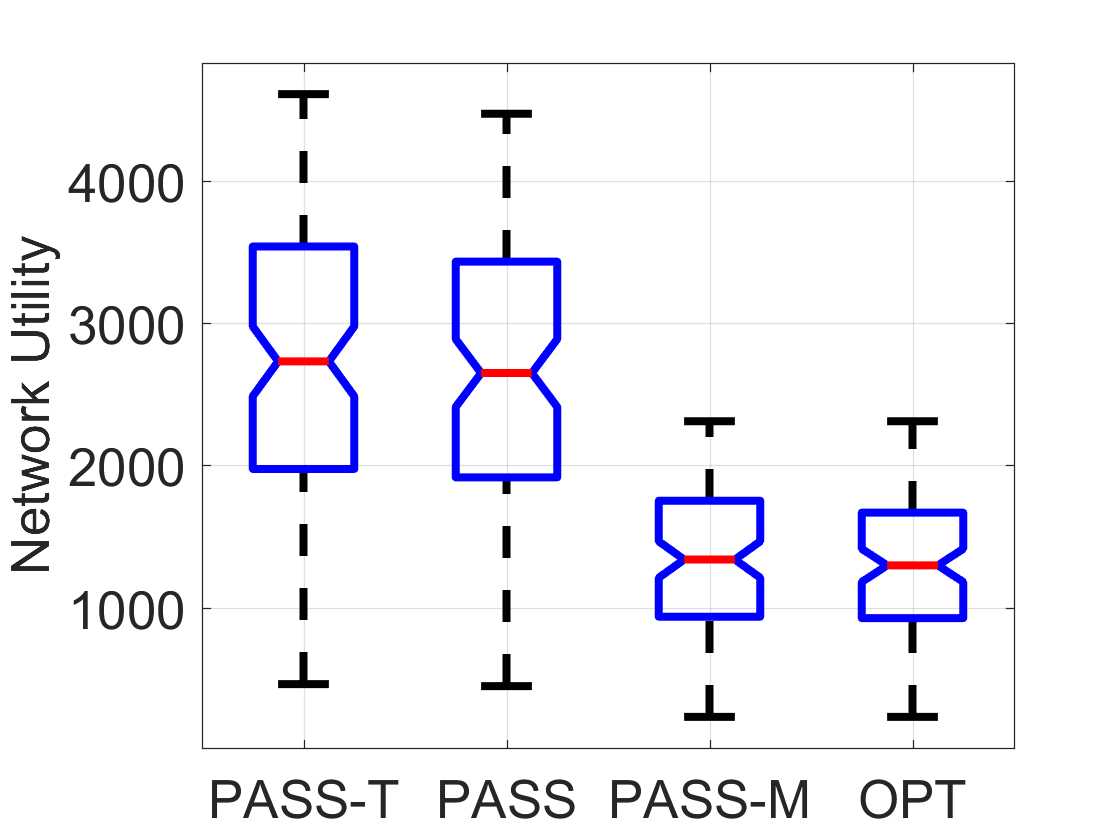}
		\label{subfig:utility-rate}}
	%\hfil
	\subfigure[Network utility increment compared to optimal, with $\epsilon=3\%$ in \textsf{PASS}.]{\includegraphics[width=1.49in]{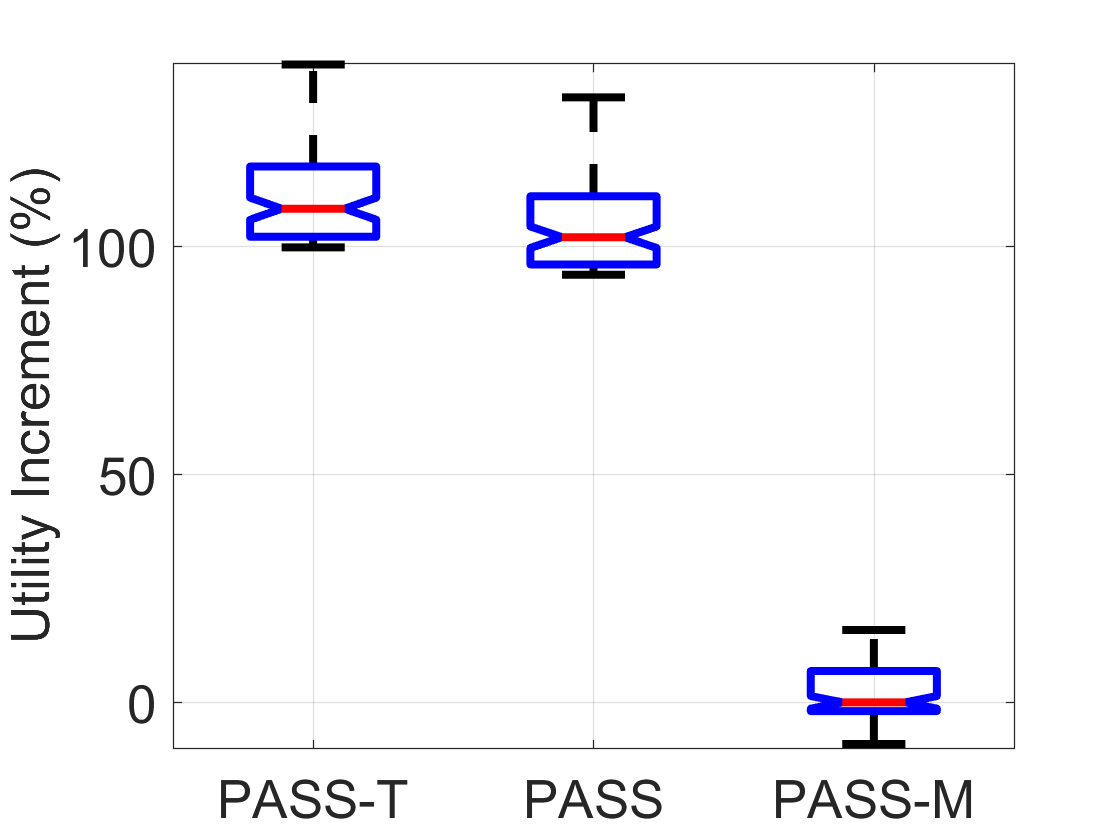}
		\label{subfig:utility-ratio}}
	\caption{Simulation results of using \textsf{PASS} to maximize network utility, with $R_1=R_2=80$ and $D_1=D_2=150$.}
	\label{fig:utility}
\end{figure}
Finally we use \textsf{PASS} to maximize aggregate user utilities, subject to link capacity constraints, maximum delay constraints, and throughput requirements (i.e., to solve \textsf{MUDT} with formula~\eqref{eqn:MUDT}). We assume the objective is~\eqref{eqn:OUDT-obj-1} where $\mathcal{U}_i^t(|f_i|)=w_i\cdot|f_i|,i=1,2$. And we assume $R_1=R_2=80$, and $D_1=D_2=150$ in the formula~\eqref{eqn:MUDT}.

We vary the weight $w_1$ (resp. $w_2$) from $1$ to $10$ with a step of $1$, thus leading to $100$ simulation instances each of which is characterized by a specific $\langle w_1,w_2\rangle,1\le w_1\le 10,1\le w_2\le 10$. For each instance, we respectively run \textsf{PASS}, \textsf{PASS-M}, \textsf{PASS-T}, and compare their solutions with the optimal. Note that we obtain the optimal solution by solving multiple-unicast flow problems in the time-expanded network, similar to Sec.~\ref{subsec:TCDM} and~\ref{subsec:DCUM}.

We present the achieved network utilities of different algorithms of the $100$ simulation instances in Fig.~\ref{subfig:utility-rate}. And in Fig.~\ref{subfig:utility-ratio}, we give the utility increment ($\%$) of our designed algorithms as compared to the optimal utility. Note that \textsf{PASS}, \textsf{PASS-M}, and \textsf{PASS-T} can obtain utilities that is strictly greater than optimal, because all of the three algorithms optimize utility subject to relaxed constraints, while the optimal utility is achieved by a feasible solution strictly meeting all the constraints.

From Fig.~\ref{fig:utility} we learn that \textsf{PASS} and \textsf{PASS-T} obtain a huge utility improvement compared to optimal (over $100\%$ more than optimal), while the utility achieved by \textsf{PASS-M} is close-to-optimal. According to Thm.~\ref{thm:sufficient-condition}, theoretically \textsf{PASS} can violate both throughput requirements and maximum delay constraints. Empirically, (i) the throughput achieved by \textsf{PASS} is $138$ (resp. $302$) in average for the first unicast (resp. second unicast), both satisfying throughput requirements $R_1=R_2=80$. (ii) The maximum delay experienced by \textsf{PASS} is $195$ (resp. $301$) in average for the first unicast (resp. second unicast), both violating maximum delay constraints $D_1=D_2=150$. But considering that video conferencing applications can accept a delay less than $400$ms~\cite{ITU}, the solution of \textsf{PASS} is acceptable. According to Thm.~\ref{thm:PASS-M}, theoretically \textsf{PASS-M} can meet maximum delay constraints while violate throughput requirements. Empirically, the throughput achieved by \textsf{PASS-M} is $71$ (resp. $154$) in average for the first unicast (resp. second unicast). It is clear that the first unicast flow violates throughput requirement. According to Thm.~\ref{thm:PASS-T}, theoretically \textsf{PASS-T} can meet throughput requirements while violate maximum delay constraints. Empirically, the maximum delay experienced by \textsf{PASS-T} is $222$ (resp. $322$) in average for the first unicast (resp. second unicast), both violating maximum delay constraints but within $400$ms that is the largest acceptable delay.     
\section{Conclusion}\label{sec:conclusion}
We consider the problem of maximizing aggregate user utilities subject to link capacity constraints, maximum delay constraints, and throughput requirements. A user's utility is a concave function of the achieved throughput or the experienced maximum delay. The problem is uniquely challenging due to the need of jointly considering maximum delay constraints and throughput requirements. We first prove that it is NP-complete either (i) to construct a feasible solution meeting all constraints, or (ii) to obtain an optimal solution after we relax maximum delay constraints or throughput requirements up to constant ratios. We then design the first polynomial-time approximation algorithm named \textsf{PASS} to obtain solutions that (i) achieve constant or problem-dependent approximation ratios, at the cost of (ii) violating maximum delay constraints or throughput requirements up to constant or problem-dependent ratios, under realistic conditions. \textsf{PASS} is practically useful since our conditions are satisfied in many popular application settings. We evaluate \textsf{PASS} empirically using extensive simulations of routing delay-critical video-conferencing traffic over Amazon EC2 datacenters. Our design leverage a new understanding between maximum-delay-aware problems and their average-delay-aware counterparts, which can be of independent interest and suggest a new avenue for solving maximum-delay-aware network optimization problems.

%We evaluate the performance of \textsf{FAST} in the multiple-unicast scenario of routing delay-critical video-conferencing traffic over inter-datacenter networks, using simulations based on the Amazon EC2 inter-datacenter topology. \textsf{FAST} can empirically reduce up to $30\%$ delay subject to throughput constraints, and increase up to $90\%$ throughput subject to delay constraints, as compared with a conceivable greedy baseline.

%\bibliographystyle{abbrv}
\bibliographystyle{ACM-Reference-Format}
\bibliography{References}
%\newpage
%\appendix \label{adx:lem-frame}
\section{Appendix}\label{adx:lem-frame}
\subsection{Proof to Lem.~\ref{lem:framework}}\label{adx:lem-frame}
\begin{proof}
	According to Algorithm~\ref{algorithm:framework}, for any $i=1,2,...,K$, $\bar{f}_i$ is obtained by iteratively deleting $\epsilon\cdot |\hat{f}_i|$ rate from $\hat{f}_i$.
	Suppose that there are in total $N_i$ iterations to get $\bar{f}_i$ by deleting rate from $\hat{f}_i$ (namely assume $N_i$ to be the number of iterations of the while-loop of line~\ref{line:while}).
	And we use $f_i^n$ to represent the flow of the unicast $i$ at the beginning of the $n$-th iteration
	(or equivalently, at the end of the $(n-1)$-th iteration). Obviously,
	$f_i^1=\hat{f}_i$,  $f_i^{N_i+1}=\bar{f}_i$. We denote $P_i^n$ as the set of of all \emph{flow-carrying} paths in flow $f_i^n$,
	and $p_i^n\in P_i^n$ as the slowest flow-carrying path in $P_i^n$. In the $n$-th iteration of the unicast $i$, \textsf{PASS} delete some rate, say $x_i^n > 0$, from $p_i^n$.
	
	Since all link delays are non-negative constants, the path delay cannot increase with reduced flow rate. Thus,
	\be
	\mathcal{M}\left(f_i^{n+1}\right) \le \mathcal{M}\left(f_i^n\right),~~\forall n=1,2,...,N_i,\forall i=1,2,...,K.
	\label{equ:proof-of-non-decreasing-of-max-delay}
	\ee
%	implying that
%	\be
%	\mathcal{M}\left(\bar{f}_i\right)\le \mathcal{M}\left(\hat{f}_i\right).
%	\label{equ:maximum-dely-decrease-SO-D}
%	\ee
	
	Considering the total delay of the unicast $i$, for any $1\le n\le N_i$, we have the following held for any $i=1,2,...,K$
	\begin{equation}\label{equ:proof-of-total-delay-relationship}
	\begin{aligned}
	& \mathcal{T}\left(f_i^n\right) =  \sum_{e\in E:e\not\in p_i^n}\left[x_i^ed_e\right] + \sum_{e\in E:e\in p_i^n}\left[x_i^ed_e\right]  \\&
	= \sum_{e\in E:e\not\in p_i^n}[x_i^ed_e] + \sum_{e\in E:e\in p_i^n}\left[ \left(x_i^e-x_i^n\right) d_e+x_i^nd_e\right]  \\&
	\overset{(a)}{=}  \sum_{e\in E:e\not\in p_i^n}[x_i^ed_e] + \sum_{e\in E:e\in p_i^n}[ (x_i^e-x_i^n) d_e] +  x_i^n\mathcal{M}(f_i^n)  \\&
	\overset{(b)}{=}  \mathcal{T}(f_i^{n+1}) + x_i^n\mathcal{M}(f_i^n)  
	\overset{(c)}{\ge}  \mathcal{T}(f_i^{n+1}) + x_i^n\mathcal{M}\left(\bar{f}_i\right).
	\end{aligned}
	\end{equation}
	In \eqref{equ:proof-of-total-delay-relationship}, equality $(a)$ holds because $\sum_{e\in p_i^n}  d_e$ is the path delay of the slowest flow-carrying path $p_i^n$.
	Equality $(b)$ holds because flow $f_i^{n+1}$ is the flow when $f_i^n$ deletes $x_i^n$ rate from path $p_i^n$. Inequality $(c)$ comes from \eqref{equ:proof-of-non-decreasing-of-max-delay} and
	$f_i^{N_i+1} = \bar{f}_i$.
	
	We then do summation for \eqref{equ:proof-of-total-delay-relationship} over $n \in [1, N_i]$, and get
	\bee
	\mathcal{T}\left[\hat{f}_i\right] & = \mathcal{T}\left(f_i^1\right) \ge  \mathcal{T}\left(f_i^{N_i+1}\right) + \left(\sum_{n=1}^{N_i} x_i^n\right) \cdot \mathcal{M}(\bar{f}_i) \nnb \\
	& = \mathcal{T}\left[\bar{f}_i\right] + \epsilon\cdot \left|\hat{f}_i\right| \cdot \mathcal{M}\left(\bar{f}_i\right), \nnb
	\eee
	which proves our lemma.
\end{proof}
\subsection{Proof to Thm.~\ref{thm:sufficient-condition}}\label{adx:condition}
\begin{proof}
	\textbf{First}, we prove the polynomial time complexity. Due to condition~\ref{ass:convex-opt}, both problem~\eqref{eqn:OUAT-T} and~\eqref{eqn:OUAT-D} can be solved in polynomial time, since (i) they are convex programs with a polynomial number of variables and a polynomial number of constraints, and (ii) convex programming problems can be solved up to an arbitrarily small additive error in polynomial time (e.g., see~\cite{convex-1,convex-2} for details). For example, the time complexity is $O(|E|^3 K^3\mathcal{L})$ where $\mathcal{L}$ is the input size of the instance of the problem~\eqref{eqn:OUAT-T} or~\eqref{eqn:OUAT-D} if they are linear programs~\cite{ye1991n3l}.
After solving the average-delay-aware problem, we get $K$ single-unicast flows each of which is defined on edges. By the classic flow decomposition technique~\cite{ford1956maximal}, we can then achieve $K$ single-unicast flows $\hat{f}=\{\hat{f}_i,i=1,2,...,K\}$ each of which is defined on paths within a time of $O(|V|^2 |E| K)$. Note that the flow decomposition outputs at most $|E|$ paths for each $\hat{f}_i$, and hence there are at most $|E|$ iterations to obtain each $\bar{f}_i$ by deleting rate from $\hat{f}_i$. Overall, Algorithm~\ref{algorithm:framework} has a polynomial time complexity that is even independent to $\epsilon$ when all conditions are satisfied.  
	
	\textbf{Second}, we prove the existence of $\bar{f}$. 
	
	(i) Suppose~\eqref{eqn:OUDT-obj-1} is the objective of the problem~\eqref{eqn:MUDT}. Because problem~\eqref{eqn:MUDT} is feasible and $f^*$ is its optimal solution, $f^*$ must satisfy all the constraints of problem~\eqref{eqn:MUDT}, implying that $f^*$ also satisfies the constraints~\eqref{eqn:OUAT-T-throughput} and~\eqref{eqn:OUAT-T-feasible} of the problem~\eqref{eqn:OUAT-T} that is the average-delay-aware counterpart of the problem~\eqref{eqn:MUDT}. Now consider that we have $\mathcal{T}(g)\le \mathcal{M}(g)\cdot |g|$ for any single-unicast flow $g$, for any $i=1,2,...,K$, the following holds
	\be
	\mathcal{T}(f_i^*)~\le~ \mathcal{M}(f_i^*)\cdot|f_i^*|~\overset{(a)}{\le}~ D_i\cdot|f_i^*|,\nnb
	\ee
	where the inequality (a) comes from that $f^*$ meets the constraints~\eqref{eqn:OUDT-delay} of the problem~\eqref{eqn:MUDT}. Therefore, $f^*$ is also a feasible solution to the problem~\eqref{eqn:OUAT-T}. Due to the existence of $f^*$, problem~\eqref{eqn:OUAT-T} must be feasible and hence Algorithm~\ref{algorithm:framework} must return a solution $\bar{f}$.
	
	(ii) Suppose~\eqref{eqn:OUDT-obj-2} is the objective of the problem~\eqref{eqn:MUDT}. Because problem~\eqref{eqn:MUDT} is feasible and $f^*$ is its optimal solution, $f^*$ must meet all the constraints of problem~\eqref{eqn:MUDT}, e.g., we have $|f_i^*|\ge R_i,\forall i=1,2,...,K$. Now we construct another network flow $f$ based on $f^*$ as follows: for each $i=1,2,...,K$, we obtain $f_i$ directly from $f_i^*$, by deleting flow rate from arbitrary flow-carrying paths of $f_i^*$ till $|f_i^*|=R_i$. The existence of $f^*$ implies the existence of $f$. For problem~\eqref{eqn:OUAT-D}, it is clear that $f$ meets the throughput requirements~\eqref{eqn:OUAT-D-throughput}. Since $f^*$ meets the constraint~\eqref{eqn:OUDT-feasible}, $f$ must satisfy the constraint~\eqref{eqn:OUAT-D-feasible}. Since we delete certain flow rate from $f_i^*$ to obtain $f_i$, it is clear that the maximum delay does not increase, i.e., we have
	\be
	\mathcal{M}(f_i)~\le~\mathcal{M}(f_i^*),~~\forall i=1,2,...,K,\label{eqn:delay-f-f*}
	\ee
	further implying the following for any $i=1,2,...,K$
	\be
	\mathcal{T}(f_i)~\le~\mathcal{M}(f_i)\cdot|f_i|~=~\mathcal{M}(f_i)\cdot R_i~\le~\mathcal{M}(f_i^*)\cdot R_i~\le~D_i\cdot R_i,\nnb
	\ee
	i.e., $f$ meets the constraints~\eqref{eqn:OUAT-D-delay}. Therefore, $f$ is a feasible solution to the problem~\eqref{eqn:OUAT-D}. Due to the existence of $f$, problem~\eqref{eqn:OUAT-D} must be feasible and hence Algorithm~\ref{algorithm:framework} must return a solution $\bar{f}$. 
	
	\textbf{Third}, we prove that $\bar{f}$ satisfies the relaxed constraints~\eqref{eqn:condition-constraints}. Suppose $\hat{f}$ is the solution to the average-delay-aware problem in line~\ref{line:original-solution}. Then clearly that $\hat{f}$ meets the following constraints:
	\bse
	\bee
	& \quad \left|\hat{f}_i\right|\ge R_i,~\forall i=1,2,...,K,\label{eqn:used-utility} \\
	& \quad \mathcal{A}\left(\hat{f}_i\right)\le D_i,~\forall i=1,2,...,K,\label{eqn:used-penalty}\\
	& \quad \hat{f}=\{\hat{f}_1,\hat{f}_2,...,\hat{f}_K\}\in \mathcal{X}.
	\eee
	\ese
	
	We know $\bar{f}_i$ is the solution by deleting a rate of $\epsilon\cdot|\hat{f}_i|$ from $\hat{f}_i$ for each $i=1,2,...,K$. It is clear that $\bar{f}$ satisfies the constraints~\eqref{eqn:relaxed-throughput} and~\eqref{eqn:relaxed-feasible}. Now we look at the constraints~\eqref{eqn:relaxed-delay}.
	
	According to our Lem.~\ref{lem:framework}, for any $i=1,2,...,K$, it holds that
	\be 
	\epsilon\cdot \left|\hat{f}_i\right|\cdot \mathcal{M}\left(\bar{f}_i\right)\le \mathcal{T}\left(\hat{f}_i\right)-\mathcal{T}\left(\bar{f}_i\right)\le \mathcal{T}\left(\hat{f}_i\right),\nnb
	\ee
	implying that $\mathcal{M}(\bar{f}_i)\le\mathcal{A}(\hat{f}_i)/\epsilon,\forall i=1,2,...,K$. Based on the satisfied constraints~\eqref{eqn:used-penalty}, we have the following for any $i=1,2,...,K$
	\be
	\mathcal{M}(\bar{f}_i)~\le~\mathcal{A}(\hat{f}_i)/\epsilon~\le~ D_i/\epsilon.\nnb
	\ee
	
	\textbf{Finally}, we prove the approximation ratio of $\bar{f}$. If~\eqref{eqn:OUDT-obj-1} is the objective of problem~\eqref{eqn:MUDT}, we have
	\bee 
	&& \sum_{i=1}^{K}\mathcal{U}_i^t\left(\left|\bar{f}_i\right|\right) =\sum_{i=1}^{K}\mathcal{U}_i^t\left((1-\epsilon)\cdot\left|\hat{f}_i\right|\right)\nnb\\
	&& \overset{(a)}{\ge}(1-\epsilon)\cdot\sum_{i=1}^{K}\mathcal{U}_i^t\left(\left|\hat{f}_i\right|\right) \overset{(b)}{\ge}(1-\epsilon)\cdot\sum_{i=1}^{K}\mathcal{U}_i^t\left(\left|f_i^*\right|\right)\nnb 
	\eee
	where the inequality (b) holds because in the second part of this proof, we have proved that $f^*$ is a feasible solution to the average-delay-aware problem~\eqref{eqn:OUAT-T}, while $\hat{f}$ is its optimal solution. Inequality (a) comes from the following inequalities for each $i=1,2,...,K$
	\bee
	&& \mathcal{U}_i^t\left((1-\epsilon)\cdot\left|\hat{f}_i\right|\right) = \mathcal{U}_i^t\left(\epsilon\cdot0+(1-\epsilon)\cdot\left|\hat{f}_i\right|\right)\nnb\\
	&& \overset{(c)}{\ge}\epsilon\cdot\mathcal{U}_i^t(0)+(1-\epsilon)\cdot\mathcal{U}_i^t\left(\left|\hat{f}_i\right|\right) \overset{(d)}{\ge} (1-\epsilon)\cdot\mathcal{U}_i^t\left(\left|\hat{f}_i\right|\right)\nnb,
	\eee
	where the inequality (c) holds due to the concavity of the function $\mathcal{U}_i^t(\cdot)$, and the inequality (d) comes from that the function $\mathcal{U}_i^t(\cdot)$ is non-negative, considering that the condition~\ref{ass:convex-opt} is satisfied.
	
	If~\eqref{eqn:OUDT-obj-2} is the objective of problem~\eqref{eqn:MUDT}, we assume $f$ is the feasible solution to the average-delay-aware problem~\eqref{eqn:OUAT-D} that is constructed from $f^*$ as discussed in the second part of this proof. Then we have 
	\bee 
	&& \sum_{i=1}^{K}\mathcal{U}_i^d\left(\mathcal{M}(\bar{f}_i)\right) \le\sum_{i=1}^{K}\mathcal{U}_i^d\left(\mathcal{A}\left(\hat{f}_i\right)/\epsilon\right)\nnb\\ && \overset{(a)}{\le}\frac{1}{\epsilon}\cdot\sum_{i=1}^{K}\mathcal{U}_i^d\left(\mathcal{A}\left(\hat{f}_i\right)\right) \overset{(b)}{\le}\frac{1}{\epsilon}\cdot\sum_{i=1}^{K}\mathcal{U}_i^d\left(\mathcal{A}\left(f_i\right)\right)\nnb\\ && \le\frac{1}{\epsilon}\cdot\sum_{i=1}^{K}\mathcal{U}_i^d\left(\mathcal{M}\left(f_i\right)\right) \overset{(c)}{\le} \frac{1}{\epsilon}\cdot\sum_{i=1}^{K}\mathcal{U}_i^d(\mathcal{M}(f_i^*)),\nnb  
	\eee
%where the inequality (a) holds due to our Lem.~\ref{lem:framework} and the non-decreasing property of the function $\mathcal{D}_0(\cdot)$. The inequality (b) comes from our condition~\eqref{ass:penalty-convex}, and the inequality (c) holds because $f^*$ is a feasible solution to $\hat{\mathcal{P}}$.
	where the inequality (a) comes from the satisfied condition~\ref{ass:penalty-convex}, the inequality (b) holds since $f$ is feasible to problem~\eqref{eqn:OUAT-D} while $\hat{f}$ is optimal to problem~\eqref{eqn:OUAT-D}, and the inequality (c) is true because of the inequality~\eqref{eqn:delay-f-f*} and the non-decreasing property of $\mathcal{U}_i^d(\cdot)$.
	%where the inequality (a) comes from our condition~\eqref{ass:cost-convex}, and the inequality (b) holds because $f^*$ is a feasible solution to $\hat{\mathcal{P}}$.
\end{proof}
%\begin{comment}
\subsection{Proof to Thm.~\ref{rmk:NP-hard}}\label{adx:rmk-hard}
\begin{proof}
\textbf{First}, we consider the following problem that is a special case of the \textsf{MUDT} with relaxed maximum delay constraints,
\bee
\max && \quad -\mathcal{M}(f_1)\nnb
\\ \mbox{s.t. }
&& \quad |f_1|\ge R_1,\nnb\\
&& \quad \mathcal{M}(f_1)\le +\infty,\nnb\\
&& \quad f=\{f_1\}\in \mathcal{X}.\nnb
\eee
It has been proved to be NP-complete to find the optimal solution to above problem (see Appendix of~\cite{misra2009polynomial}).

\textbf{Second}, we consider the following problem that is a special case of the \textsf{MUDT} with relaxed throughput requirements,
\bee
\max && \quad |f_1|\nnb
\\ \mbox{s.t. }
&& \quad |f_1|\ge 0,\nnb\\
&& \quad \mathcal{M}(f_1)\le D_1,\nnb\\
&& \quad f=\{f_1\}\in \mathcal{X}.\nnb
\eee
Follow a similar proof as that in the Appendix of~\cite{misra2009polynomial}, it can be proved to be NP-complete to find the optimal solution to the aforementioned problem.

\textbf{Third}, also following a similar proof as that in the Appendix of~\cite{misra2009polynomial}, it can be proved that it is NP-complete even to construct a feasible solution to the following problem that is a special case of our \textsf{MUDT}, strictly meeting all constraints
\bse
\bee
\max && \quad \mathcal{U}_1^t(|f_1|)\nnb
\\ \mbox{s.t. }
&& \quad |f_1|\ge R_1,\nnb\\
&& \quad \mathcal{M}(f_1)\le D_1,\nnb\\
&& \quad f=\{f_1\}\in \mathcal{X},\nnb
\eee
\ese
where $\mathcal{U}_1^t(|f_1|)=1$ which is a constant.
\end{proof}
\subsection{Proof to Thm.~\ref{thm:PASS-M}}\label{adx:PASS-M}
\begin{proof}
	\textbf{First}, due to the same proof to Thm.~\ref{thm:sufficient-condition}, Algorithm~\ref{alg:PASS-M} has a polynomial time complexity, and must give a solution $\bar{f}$.
	
	\textbf{Second}, it is straightforward that constraints~\eqref{eqn:PASS-M-delay} and~\eqref{eqn:PASS-M-feasible} are met. Now let us denote $(|\hat{f}_i|-|\bar{f}_i|)/|\hat{f}_i|$ as $\epsilon_i$. Thus $\epsilon_{\min}\le\epsilon_i\le\epsilon_{\max}$ for any $i=1,2,...,K$, implying the following
	\be
	\left|\bar{f}_i\right|~=~(1-\epsilon_i)\cdot\left|\hat{f}_i\right|~\ge~(1-\epsilon_{\max})\cdot\left|\hat{f}_i\right|,~~\forall i=1,2,...,K,\nnb
	\ee
	i.e., the constraints~\eqref{eqn:PASS-M-delay} are satisfied.
	
	\textbf{Third}, following the same proof as to Thm.~\ref{thm:sufficient-condition}, the approximation ratio~\eqref{eqn:PASS-M-T-ratio} can be proved.
	
	As for the approximation ratio~\eqref{eqn:PASS-M-D-ratio}, let as assume $\tilde{f}$ to be the solution where for each $i=1,2,...,K$, we delete $\epsilon_{\min}|\hat{f}_i|$ rate from the slowest flow-carrying paths of $\hat{f}_i$ to obtain $\tilde{f}_i$. It is clear that
	\be
	\mathcal{M}(\bar{f}_i)~\le~\mathcal{M}(\tilde{f}_i),~~\forall i=1,2,...,K,\nnb
	\ee
	because both $\bar{f}_i$ and $\tilde{f}_i$ are flows after we delete rates from the slowest flow-carrying paths of $\hat{f}_i$, but the amount of deleted rate to obtain  $\bar{f}_i$ is no smaller than the amount of deleted rate to obtain $\tilde{f}_i$, for each $i=1,2,...,K$. Therefore, we have the following
	\be 
	\sum_{i=1}^{K}\mathcal{U}_i^d\left(\mathcal{M}\left(\bar{f}_i\right)\right)\le  \sum_{i=1}^{K}\mathcal{U}_i^d\left(\mathcal{M}\left(\tilde{f}_i\right)\right) \overset{(a)}{\le}\frac{1}{\epsilon_{\min}}\cdot\sum_{i=1}^{K}\mathcal{U}_i^d\left(\mathcal{M}\left(f_i^*\right)\right),\nnb
	\ee
	where the inequality (a) comes from our Thm.~\ref{thm:sufficient-condition}, since $\tilde{f}$ is also the solution returned if we use Algorithm~\ref{algorithm:framework} with $\epsilon=\epsilon_{\min}$ to solve the problem~\eqref{eqn:MUDT}.  
\end{proof}
\subsection{Proof to Thm.~\ref{thm:PASS-T}}\label{adx:PASS-T}
\begin{proof}
	Same to the proof as that of Thm.~\ref{thm:sufficient-condition}, it holds that \textsf{PASS-T} must return a solution $\bar{f}$ in a polynomial time, meeting the constraints~\eqref{eqn:PASS-T-throughput},~\eqref{eqn:PASS-T-feasible}, and providing the approximation ratio~\eqref{eqn:PASS-T-T-ratio}.
	
	Because that $\bar{g}$ is the solution of \textsf{PASS}, we have
	\bee
	&& \mathcal{M}\left(\bar{g}_i\right)\le D_i/\epsilon,~~\forall i=1,2,...,K,\label{eqn:use-1}\\
	&& \sum_{i=1}^{K}\mathcal{U}_i^d\left(\mathcal{M}\left(\bar{g}_i\right)\right)\le \frac{1}{\epsilon}\cdot\sum_{i=1}^{K}\mathcal{U}_i^d\left(\mathcal{M}\left(f_i^*\right)\right).\label{eqn:use-2}
	\eee
	
	According to the definition of $\lambda$, we have
	\be
	\mathcal{M}(\bar{f}_i)\le\lambda\cdot\mathcal{M}(\bar{g}_i),~~\forall i=1,2,...,K,\nnb
	\ee
	implying the following considering the inequality~\eqref{eqn:use-1}
	\be
	\mathcal{M}(\bar{f}_i)\le \lambda\cdot D_i/\epsilon,~~\forall i=1,2,...,K,\nnb\\
	\ee
	i.e., $\bar{f}$ satisfies the constraints~\eqref{eqn:PASS-T-delay}. We further have
	\bee
	&& \sum_{i=1}^{K}\mathcal{U}_i^d\left(\mathcal{M}(\bar{f}_i)\right) \le \sum_{i=1}^{K}\mathcal{U}_i^d\left(\lambda\cdot\mathcal{M}(\bar{g}_i)\right)\nnb\\
	&& \overset{(a)}{\le} \lambda\cdot\sum_{i=1}^{K}\mathcal{U}_i^d\left(\mathcal{M}(\bar{g}_i)\right) \overset{(b)}{\le} \frac{\lambda}{\epsilon}\cdot\sum_{i=1}^{K}\mathcal{U}_i^d\left(\mathcal{M}(f_i^*)\right),\nnb
	\eee
	where the inequality (a) comes from the satisfied condition~\ref{ass:penalty-convex}, and the inequality (b) holds due to the inequality~\eqref{eqn:use-2}. Thus the approximation ratio~\eqref{eqn:PASS-T-D-ratio} holds.
\end{proof}
%\end{comment} 

\end{document}